\documentclass[12pt]{report}
\usepackage{geometry}
\usepackage{fancyhdr}
\usepackage{titlesec}
\usepackage{lipsum}
\usepackage{graphicx}
\usepackage{setspace}
\usepackage{amsmath}
\usepackage{amsthm}
\usepackage{amssymb}
\usepackage{bm}
\usepackage{tikz-cd}
\usepackage{array}
\usepackage{tabularx}
\usepackage{titletoc}
\usepackage{booktabs}
\usepackage{float}
\usepackage{fontspec}
\usepackage{listings}
\newfontfamily\juliamono{JuliaMono-Regular.ttf}[
  ItalicFont = JuliaMono-RegularItalic.ttf,
  BoldFont = JuliaMono-Bold.ttf,
  BoldItalicFont = JuliaMono-BoldItalic.ttf,
  Scale = MatchLowercase,
]

\usepackage[backend=biber, style=ieee, sorting=none]{biblatex}
\usepackage{aliascnt}
\usepackage[nameinlink]{cleveref}

\renewcommand{\headrulewidth}{0pt}
\fancypagestyle{plain}{
  \fancyhf{} 
  \fancyfoot[R]{\thepage} 
  \renewcommand{\headrulewidth}{0pt}}

\makeatletter
\def\@makechapterhead#1{%
  { \parindent \z@ \normalfont
    \ifnum \c@secnumdepth >\m@ne
        \thechapter \space - \space        
    \fi
    \interlinepenalty\@M
    \centering \MakeUppercase{#1}\par\nobreak
    \vskip 12\p@
  }
}
\def\@makeschapterhead#1{%
  {\parindent \z@
    \normalfont
    \interlinepenalty\@M
    \centering \MakeUppercase{#1}\par\nobreak
    \vskip 12\p@
  }
}
\def\@makesection#1{%
  {\parindent \z@
    \normalfont
    \interlinepenalty\@M
    \centering \MakeUppercase{#1}\par\nobreak
    \vskip 12\p@
  }
}
\makeatother

\titlecontents{chapter}
[0em] %
{\normalfont}
{\contentslabel{0pt}\hspace{0.7cm}\MakeUppercase}
{\contentslabel{0pt}\MakeUppercase}
{\hfill\contentspage}[]

\titlecontents{section}
[1.8em] %
{\normalfont}
{\contentslabel{0pt}\hspace{1cm}\MakeUppercase}
{\contentslabel{0pt}\MakeUppercase}
{\hfill\contentspage}[]

\titlecontents{subsection}
[4.1em] %
{\normalfont}
{\contentslabel{0pt}\hspace{1.3cm}\MakeUppercase}
{\contentslabel{0pt}\MakeUppercase}
{\hfill\contentspage}[]

\titleformat{\section}[block]{\normalfont}{}{0cm}{\thesection \space - \space \MakeUppercase}
\titleformat{name=\section,numberless}[block]{\normalfont}{}{0cm}{\MakeUppercase}
\titlespacing{\section}{0cm}{0cm}{0cm}

\titleformat{\subsection}[block]{\normalfont}{}{0cm}{\thesubsection \space - \space \MakeUppercase}
\titleformat{name=\subsection,numberless}[runin]{\normalfont}{}{12pt}{\MakeUppercase}[\space - \space\space]
\titlespacing{\subsection}{0cm}{0cm}{0cm}

\newtheorem{theorem}{Theorem}[section]
\newaliascnt{corollary}{theorem}
\newtheorem{corollary}[corollary]{Corollary}
\aliascntresetthe{corollary}
\newaliascnt{lemma}{theorem}
\newtheorem{lemma}[lemma]{Lemma}
\aliascntresetthe{lemma}

\crefname{theorem}{theorem}{theorems}
\Crefname{theorem}{Theorem}{Theorems}
\crefname{corollary}{corollary}{corollaries}
\Crefname{corollary}{Corollary}{Corollaries}
\crefname{lemma}{lemma}{lemmas}
\Crefname{lemma}{Lemma}{Lemmas}

\theoremstyle{definition}
\newtheorem{definition}{Definition}[section]

\newcommand{\prd}[2]{\prod\nolimits_{(#1)}\,#2}
\newcommand{\sm}[2]{\sum\nolimits_{(#1)}\,#2}

\newcommand{\RR}{\mathbb{R}}
\newcommand{\QQ}{\mathbb{Q}}
\newcommand{\Qp}{\mathbb{Q}_{+}}

\newcommand{\close}[1]{\sim_{#1}}
\newcommand{\closesym}{\mathord\sim}
\newcommand{\rational}{\mathsf{rational}}
\newcommand{\limit}{\mathsf{limit}}
\newcommand{\cpath}{\mathsf{path}}
\newcommand{\transport}{\mathsf{transport}}
\newcommand{\induction}{\mathsf{induction}}
\newcommand{\recursion}{\mathsf{recursion}}
\newcommand{\IsCauchy}{\mathsf{IsCauchy}}
\NewDocumentCommand{\dclose}{m o}{%
  \mathrel{\IfNoValueTF{#2}{\frown_{#1}}{\frown_{#1}^{#2}}}%
}
\newcommand{\dclosesym}{\mathord\frown}
\newcommand{\apart}{\mathrel{\#}}
\newcommand{\unittype}{\mathbf{1}}

\newcommand{\trunc}[2]{\mathopen{}\left\Vert #2\right\Vert_{#1}\mathclose{}}
\newcommand{\bracket}[1]{\trunc{}{#1}}

\begin{document}

\renewcommand{\contentsname}{TABLE OF CONTENTS}

\begin{titlepage}
\begin{center}
    Formalizing the Real Numbers in \\
    \bigskip
    Homotopy Type Theory with Cubical Agda \\
    \bigskip
    \bigskip
    \bigskip
    by\\
    \bigskip
    Jackson Brough\\
    \bigskip
    \bigskip
    \bigskip
    A Senior Honors Thesis Submitted to the Faculty of\\
    The University of Utah\\
    In Partial Fulfillment of the Requirements for the\\
    \bigskip
    Honors Degree in Bachelor of Science\\
    \bigskip
    \bigskip
    \bigskip
    In \\
    \bigskip
    Computer Engineering
    \bigskip
    \bigskip
\end{center}
Approved:
\begin{center}
  \begin{tabularx}{\textwidth}{XX}
    \bigskip \bigskip \bigskip \\
    \rule{6cm}{0.4pt} & \rule{6cm}{0.4pt} \\
    Ben Greenman & Neal Patwari \\
    Faculty Thesis Mentor & Departmental Honors Liaison \\
    \bigskip \bigskip \bigskip \\
    \rule{6cm}{0.4pt} & \rule{6cm}{0.4pt}\\
    Hanseup Kim & Monisha Pasupathi \\
    Chair, Department of ECE & Dean, Honors College \\
  \end{tabularx}
\end{center}

\begin{center}
    \bigskip
    \bigskip
    \bigskip
    \bigskip

    April 2026\\
    Copyright \textcopyright\ 2026\\
    All Rights Reserved
\end{center}
\end{titlepage}

\setstretch{2.0}

\clearpage
\pagenumbering{roman}
\setcounter{page}{2} 

\chapter*{ABSTRACT}
\addcontentsline{toc}{chapter}{ABSTRACT}
\bigskip

Real numbers in constructive mathematics have always seemed to require
compromises of one form or another. Classical proofs of Cauchy completeness
require countable choice, Bishop's setoid construction introduces persistent
bookkeeping overhead on every definition and theorem, and Dedekind cuts force
cumbersome universe-level tracking in predicative type theory. The Homotopy Type
Theory (HoTT) book presents an alternative construction of the Cauchy real
numbers as a higher inductive-inductive type family, avoiding all three
compromises. We formalize the HoTT book reals in Cubical Agda, a proof assistant
whose native support for higher inductive types allows the construction to be
expressed directly. The code type-checks without postulates or holes, providing
a foundation for further machine-assisted work in constructive analysis.

\newpage

{
\let\clearpage\relax
\tableofcontents
}

\clearpage
\pagenumbering{arabic}

\pagestyle{fancy}
\fancyhf{}
\renewcommand{\headrulewidth}{0pt}
\fancyhead[R]{\thepage}

\fancypagestyle{plain}{
  \fancyhf{} 
  \fancyhead[R]{\thepage} 
  \renewcommand{\headrulewidth}{0pt}}

\chapter{INTRODUCTION}

Constructive mathematics endows mathematical proofs with a computational
interpretation. When a constructive proof claims ``given $\varepsilon$, there
exists a $\delta$,'' it is expected to provide a method for obtaining such a
$\delta$. Given a specific $\varepsilon$, a reader should in principle be able
to follow the proof step by step to determine a corresponding $\delta$.
By contrast, a classical proof is only required to show that a $\delta$
cannot fail to exist.

Constructive real analysis is attractive because it promises to narrow the gap
between abstract existence claims and the explicit procedures needed to realize
them in practice.
Yet, giving a
constructive account of real analysis becomes challenging almost immediately,
because unlike elements of discrete number systems, arbitrary real numbers
cannot in general be specified by finite data; in particular, equality of
arbitrary real numbers is not effectively decidable. Prior constructions of
the real numbers have always required compromises of one kind or another
\cite[p.~375]{univalentfoundationsprogram2013}.

To make the tension precise, consider the classical Cauchy construction of the
real numbers. Classically, the Cauchy reals are obtained by taking the quotient
of the Cauchy sequences of rational numbers under the equivalence relation that
identifies sequences which are eventually arbitrarily close. The crucial step is
to show Cauchy completeness: every Cauchy sequence of reals should have a limit
in the reals. However, the classical proof of Cauchy completeness for the
quotient fails in a constructive setting. We are given only a sequence of
equivalence classes, but to build a representative for its limit we must choose
a representative for each term of that sequence. Doing so requires the axiom of
countable choice: for each term of the sequence we know the equivalence class is
nonempty, but we must produce a function that selects a particular
representative for each term. Constructively, the equivalence classes do not
single out any particular representative, so no such function is available in
general.

In his 1967 book \emph{Foundations of Constructive Analysis}, Errett Bishop
showed that a substantial portion of classical analysis could be recast
constructively~\cite{bishop1967}. His aim was ``to give a numerical meaning to
as much as possible of classical abstract analysis''~\cite[p.~3]{bishop1967}.
Bishop's solution to the problem of Cauchy completeness is to avoid taking the
quotient and to work instead with the setoid of Cauchy sequences of rational
numbers, where a setoid is a set equipped with a canonical equivalence
relation. This avoids the need for countable choice, but it introduces a
persistent bookkeeping burden: because compatibility with the equivalence
relation is not automatic, the same kinds of congruence proofs must be supplied
repeatedly, often for routine constructions. Taken to its limit, this approach
would force all of abstract algebra to be reformulated in terms of setoids in
order to apply results to the reals, as observed by Gilbert~\cite{gilbert2017}.

Another alternative is to work with Dedekind reals, but these can be awkward in
predicative type theory because their construction involves Dedekind
cuts---subsets valued in a universe of propositions. The real numbers defined in
this way thus live one universe level above their cuts, and a property of the
reals one level higher still. Tracking these varying universe levels throughout
a development becomes cumbersome~\cite[p.~377]{univalentfoundationsprogram2013}.

In the final chapter of the \emph{Homotopy Type Theory} (HoTT) book, the authors
present a new formulation of the Cauchy reals as a higher inductive-inductive
type family~\cite{univalentfoundationsprogram2013}. We refer to this
construction as \emph{the HoTT book reals}, to distinguish it from the naive
quotient of Cauchy sequences discussed previously. This construction is
compelling because it avoids the compromises of the previous approaches. First,
it builds Cauchy completeness directly into the definition by including a limit
constructor for Cauchy approximations valued in the reals. The accompanying
closeness relation is defined simultaneously with the type of reals itself, so
this construction is not circular. This avoids the need for countable
choice. Second, unlike in Bishop's work, it aligns the identity type of the
reals with their intended notion of equality: the HoTT book reals come equipped
with a path constructor identifying any two reals that are $\varepsilon$-close
for every rational $\varepsilon > 0$. Finally, because this construction is not
defined in terms of proposition-valued Dedekind cuts, it avoids the associated
burden of tracking universe levels in predicative type theory.

Gilbert formalized the HoTT book reals in a modified version of the Rocq proof
assistant and even generalized the construction to premetric spaces, following
earlier work of O'Connor~\cite{gilbert2017,oconnor2007}. However, Rocq by
default lacks support for both inductive-inductive type families and higher
inductive types. Gilbert worked with an experimental branch of Rocq that added
inductive-inductive type families, but higher inductive types were still
unavailable, so the path constructors for the reals and the accompanying
closeness relation had to be postulated. As a result, the formalization could
not be type checked in mainline Rocq, and the corresponding path constructors
did not enjoy native computation rules.

In this thesis, we present a formalization of the HoTT book reals in Cubical
Agda, whose native support for higher inductive types allows the construction to
be expressed directly, without postulates. Cubical Agda implements cubical type
theory, which gives a constructive interpretation of both higher inductive types
and univalence, making it a natural setting for this
work~\cite{vezzosimortbergabel2019}. Our formalization follows Section 11.3 of
the HoTT book. The code is open source and available at
\url{https://github.com/utahplt/hott-reals}.

Chapter 2 presents the structure of the formalization in four stages. It begins
with the higher inductive-inductive definition of the HoTT book reals and
several variations of induction and recursion principles used throughout the
development. It then develops a collection of results concerning the closeness
relation and an alternative characterization thereof. From there, the chapter
establishes results on lifting Lipschitz and nonexpanding maps to the reals and
finally develops the algebraic and order structure of the HoTT book reals,
culminating in the proof that reals form an Archimedean ordered field.

Chapter 3 reflects on what the Agda code revealed about the mathematical
construction in Chapter 2. In several places, the process of formalization
sharpened our understanding of the informal presentation. We articulated the
hypotheses needed to extend multi-variable identities from the rationals to the
reals, repaired a naive transcription of the enhanced recursion principle to a
form stronger than an informal reading suggests, and recognized the structural
necessity of the alternative characterization of closeness. The chapter also
includes a demonstration that real arithmetic computes definitionally on
rational inputs, an overview of the codebase organization, and a discussion of
our experimentation with Claude Code during development.

Chapter 4 presents related work.
Chapter 5 outlines future work.
Chapter 6 concludes.

The Cubical Agda code typechecks without postulates or holes using the latest
standard library release, and is available for other researchers to build
on. The development clarifies how the HoTT book reals behave in a proof
assistant with native support for higher inductive types. It provides a
foundation for further machine-assisted work in constructive analysis.

\chapter{Real numbers as a higher inductive type}
\label{chapter:reals-math}

This chapter presents the mathematical development of the HoTT book reals
($\RR$), beginning with their higher inductive-inductive definition and
ultimately concluding that they form an Archimedean ordered field. We follow the
notation and conventions of the HoTT book \cite{univalentfoundationsprogram2013}
and assume familiarity with basic Homotopy Type Theory. Readers less familiar
with HoTT should still be able to follow the main ideas by consulting the HoTT
book occasionally when unfamiliar concepts or notation arise.

The chapter is organized as follows. In
\Cref{section:definition-induction-recursion}, we introduce the definition of
the HoTT book reals and specify their induction and recursion principles. We
then develop the basic theory of the closeness relation in
\Cref{section:closeness}, and use it to state extension principles for Lipschitz
and non-expanding maps in \Cref{section:lifting}. Finally, in
\Cref{section:algebra-order}, we assemble the resulting algebraic and
order-theoretic structure. The first three sections follow the HoTT book
closely. In the final section, we continue to use the HoTT book for the basic
algebraic constructions, but follow Gilbert \cite{gilbert2017} for the treatment
of strict order and for the construction of multiplication and reciprocal, with
Kraus's theorem on maps out of propositional truncations into sets providing a
convenient formulation of one local-to-global step in the latter construction
\cite{kraus2015}.

By default, we do not reproduce proofs already given in the HoTT book or
Gilbert. We repeat an existing proof only when its content contributes
meaningfully to the conceptual narrative, or when our proof strategy differs
from the one given in those sources.

\section{Definition, Induction, and Recursion}
\label{section:definition-induction-recursion}

The HoTT book reals are an adaptation of the classical Cauchy construction, in
which the real numbers are obtained by adjoining limits to the rationals. The
completion of a general metric space is obtained by adjoining limits of
sequences satisfying the Cauchy condition, that is, sequences whose terms
eventually become arbitrarily close. However, metric spaces presuppose a notion
of distance valued in the real numbers, so proceeding in this way would be
circular here, since we are in the process of constructing the reals
themselves. The key observation is that the full structure of a metric is not
required to formulate the Cauchy condition; it suffices to express approximate
closeness between reals. In the HoTT book construction, this is captured by a
family of binary closeness relations on the reals, indexed by positive
rationals.

The main difficulty is that the definitions of the reals and the closeness
relation depend on each other. The limit constructor for the reals takes as
input a family of real approximations satisfying the Cauchy condition formulated
in terms of closeness. Conversely, the closeness relation is defined by cases on
the point constructors for the reals: rational-rational, rational-limit,
limit-rational, and limit-limit. This kind of mutual dependence is handled by
the framework of inductive-inductive definitions, which allow the simultaneous
definition of a type $A$ together with a type family $B : A \to \mathcal{U}$
over $A$~\cite{forsberg2013}. In the present case, this simultaneously defined
family takes the form of a ternary relation: for reals $u, v : \RR$ and a
positive rational $\varepsilon$, the type $u \close{\varepsilon} v$ expresses
that $u$ and $v$ are $\varepsilon$-close.

In addition to adjoining limits, we must also ensure that arbitrarily close
reals are identified. Specifically, reals which are $\varepsilon$-close for
every positive rational $\varepsilon$ ought to be treated as
indistinguishable. This requirement is built into the definition via a path
constructor that provides a path between any two reals satisfying this
condition. For this reason, the HoTT book reals and the associated closeness
relation are defined as a higher inductive-inductive type family.

To express the Cauchy condition in terms of the closeness relation, the HoTT
book uses Cauchy approximations rather than traditional Cauchy sequences. A
\textbf{Cauchy approximation} is a map $x : \Qp \to \RR$ satisfying
$x_{\varepsilon} \close{\varepsilon + \delta} x_{\delta}$ for all positive
rationals $\varepsilon$ and $\delta$. Once closeness is indexed by positive
rationals, it is natural for families of approximations to share the same
index. Intuitively, a Cauchy approximation assigns, for each requested precision
$\varepsilon > 0$, an approximate value of the intended limit. In contrast, an
ordinary Cauchy sequence $\mathbb{N} \to \RR$ does not by itself indicate which
term to use for a requested precision; a modulus of convergence is also
needed. Cauchy approximations therefore package a Cauchy sequence together with
its modulus of convergence into a form suited to the closeness-based formulation
of completion.

With this motivation in place, we can state the higher inductive-inductive
definition of the HoTT book reals and their closeness relation.

\begin{definition}[{\textcite[Definition 11.3.2]{univalentfoundationsprogram2013}}]
  \label{def:hott-reals}
  The type $\RR$ of \textbf{HoTT book reals}, together with a type family
  $\closesym : \Qp \to \RR \to \RR \to \mathcal{U}$ referred to as the
  \textbf{closeness relation}, are defined simultaneously as the following
  higher inductive-inductive type family. The type $\RR$ comes equipped with the
  following constructors:
  \begin{itemize}
    \item For each rational $q : \QQ$, there is an element $\rational(q) : \RR$.
    \item For each map $x : \Qp \to \RR$ equipped with an element of type
      \[
        \IsCauchy(x) := \prd{\varepsilon, \delta : \Qp} x_{\varepsilon} \close{\varepsilon + \delta} x_{\delta},
      \]
      there is an element $\limit(x) : \RR$. We call $x$ a \textbf{Cauchy
        approximation}.
    \item For each $u, v : \RR$ such that $u \close{\varepsilon} v$ for all
      $\varepsilon : \Qp$, there is a path $\cpath(u, v) : u = v$.
  \end{itemize}
  The closeness relation $\closesym : \Qp \to \RR \to \RR \to \mathcal{U}$ comes
  equipped with the following constructors:
  \begin{itemize}
  \item For all rationals $q, r, \varepsilon$ with $\varepsilon > 0$, if
    $- \varepsilon < q - r < \varepsilon$ then
    $\rational(q) \close{\varepsilon} \rational(r)$.
  \item For all rationals $q, \varepsilon, \delta$ with $\varepsilon > 0$,
    $\delta > 0$, and $\varepsilon - \delta > 0$ and all Cauchy approximations
    $y : \Qp \to \RR$, if $\rational(q) \close{\varepsilon - \delta} y_{\delta}$
    then $\rational(q) \close{\varepsilon} \limit(y)$.
  \item For all Cauchy approximations $x : \Qp \to \RR$ and all rationals
    $r, \varepsilon, \delta$ with $\varepsilon > 0$, $\delta > 0$, and
    $\varepsilon - \delta > 0$, if
    $x_{\delta} \close{\varepsilon - \delta} \rational(r)$ then
    $\limit(x) \close{\varepsilon} \rational(r)$.
  \item For all Cauchy approximations $x, y$ and all rationals
    $\varepsilon, \delta, \eta$ with $\varepsilon, \delta, \eta > 0$ and
    $\varepsilon - (\delta + \eta) > 0$, if
    $x_{\delta} \close{\varepsilon - (\delta + \eta)} y_{\eta}$ then
    $\limit(x) \close{\varepsilon} \limit(y)$.
  \item Given reals $u, v$ and a rational $\varepsilon > 0$, if
    $\varphi, \varphi' : u \close{\varepsilon} v$ then there is a path
    $\varphi = \varphi'$.
  \end{itemize}
\end{definition}

The $\rational$ constructor of $\RR$ embeds each rational number as a real
number. The $\limit$ constructor asserts that every Cauchy approximation
$x : \Qp \to \RR$ determines a real number $\limit(x)$, which we refer to as the
limit of $x$. The remaining constructor, discussed below, identifies reals that
are arbitrarily close.

Heuristically, $x_{\varepsilon}$ is intended to be an approximation of
$\limit(x)$ corresponding to the error tolerance $\varepsilon$. Assuming the
limit exists, informal reasoning with the triangle inequality suggests that
$x_{\varepsilon}$ and $x_{\delta}$ are $(\varepsilon + \delta)$-close. The
$\IsCauchy$ condition can be understood as internalizing this reasoning while
avoiding any circular reference to the limit itself. A formal definition of the
triangle inequality in the context of the closeness relation is given in
\Cref{definition:triangle-inequality}.

In the Agda formalization, a proof $\varphi : \IsCauchy(x)$ is tracked
explicitly, but in our informal presentation here, we follow
\cite{univalentfoundationsprogram2013} and abuse notation slightly by writing
$\limit(x)$ instead of $\limit(x, \varphi)$.

The path constructor enforces that any two reals that are \(\varepsilon\)-close
for every positive rational \(\varepsilon\) are equal. In the HoTT book, this
property is called \textbf{separatedness}
\cite[\S11.3.2]{univalentfoundationsprogram2013}. The terminology is analogous
to the Hausdorff condition in topology, which classically says that distinct
points can be separated by disjoint neighborhoods. For the real numbers, this
matches the familiar metric intuition that distinct points are separated by a
positive distance. Constructively, we use the corresponding positive
formulation: if two reals are arbitrarily close, then they are equal.

The first constructor for closeness lifts the usual condition for
$\varepsilon$-closeness for rational numbers to the reals. The next three
constructors describe how closeness interacts with limits and can be understood
using the heuristic explanation above: we think of $x_{\varepsilon}$ as an
approximation to $\limit(x)$ associated with the error tolerance $\varepsilon$,
and then reason informally via the triangle inequality. The final constructor
asserts that closeness is a proposition, that is, given $u, v : \RR$ and
$\varepsilon > 0$, any two proofs of $u \close{\varepsilon} v$ are equal.

The use of a higher inductive-inductive type family, together with the large
number of constructors in \Cref{def:hott-reals}, makes the induction principle
quite involved. Because $\RR$ and $\closesym$ are defined mutually, it is not
sufficient in general to work with a type family only over $\RR$ or only over
$\closesym$, though we will consider these as special cases shortly.

We will work up to the full induction principle following
Forsberg~\cite[\S3.2.5]{forsberg2013}.  Forsberg makes a distinction between
``simple'' and ``general'' elimination principles for inductive-inductively
defined types. Suppose $A : \mathcal{U}$ and $B : A \to \mathcal{U}$ are
inductive-inductively defined types. Simple elimination rules use motives of the
form
\begin{align*}
  P & : A \to \mathcal{U}, \\
  Q & : \prd{x : A} B(x) \to \mathcal{U}
\end{align*}
and therefore have the form
\begin{align*}
  \induction'_{A} & : \cdots \to \prd{x : A} P(x), \\
  \induction'_{B} & : \cdots \to \prd{x : A} \prd{y : B(x)} Q(x, y).
\end{align*}
In contrast, general elimination principles use motives of the form
\begin{align*}
  P & : A \to \mathcal{U}, \\
  Q & : \prd{x : A} \prd{y : B(x)} P(x) \to \mathcal{U}
\end{align*}
and therefore have the form
\begin{align*}
  \induction_{A} & : \cdots \to \prd{x : A} P(x), \\
  \induction_{B} & : \cdots \to \prd{x : A} \prd{y : B(x)} Q(x, y, \induction_{A}(\ldots, x)).
\end{align*}
With a general induction principle, the motive $Q$ may depend on the result of
induction for $P$. When constructing data over an indexed family, we often need
access not only to the index, but also to the data already assigned to that
index by induction.

The present case is slightly more elaborate, since the closeness relation is
indexed by two reals and a positive rational~\cite[\S6.2]{forsberg2013}.
The codomain families in our setting therefore take
the form
\begin{align*}
  A & : \RR \to \mathcal{U}, \\
  B & : \prd{u, v : \RR} A(u) \to A(v) \to \prd{\varepsilon : \Qp} u \close{\varepsilon} v \to \mathcal{U}.
\end{align*}
The family $B$ may be understood as an $\varepsilon$-indexed binary relation
between the fibers $A(u)$ and $A(v)$, defined only when $u, v : \RR$ are known
to be $\varepsilon$-close. Accordingly, we usually write this type using infix
notation as $(u, a) \dclose{\varepsilon}[\varphi] (v, b)$, where $a : A(u)$,
$b : A(v)$, and $\varphi : u \close{\varepsilon} v$. Because closeness is a mere
relation, the witness $\varphi$ is often irrelevant, and the endpoints $u$ and
$v$ are usually clear from context, so we will often abbreviate this as
$a \dclose{\varepsilon} b$.

\begin{definition}[{\textcite[\S~11.3.2]{univalentfoundationsprogram2013}}]
  \label{def:general-induction}
  Assume codomain families $A$ and $\dclosesym$ as above. The general induction
  principle for the HoTT book reals,
  \textbf{$(\RR, \closesym)$-induction}, asserts that to construct sections of
  the families $A$ and $B$, it suffices to specify data corresponding to each
  constructor of the higher inductive-inductive definition. Concretely, the
  required hypotheses are as follows:

  \begin{itemize}
    \item For each $q : \QQ$, an element $f_{\rational}(q) : A(\rational(q))$.
    \item For each Cauchy approximation $x : \Qp \to \RR$ and each
      $a : \prd{\varepsilon : \Qp} A(x_{\varepsilon})$ satisfying
      \begin{equation}
        \label{eq:dependent-cauchy-approximation}
        \prd{\varepsilon, \delta : \Qp} (x_{\varepsilon}, a_{\varepsilon})
        \dclose{\varepsilon + \delta} (x_{\delta}, a_{\delta}),
      \end{equation}
      an element $f_{\limit}(x, a) : A(\limit(x))$. Whenever such an $a$
      satisfies \labelcref{eq:dependent-cauchy-approximation}, we refer to it as
      a \textbf{dependent Cauchy approximation} over $x$. As with $\limit(x)$,
      we typically suppress the witness that $x$ is a Cauchy approximation and
      that $a$ is a dependent Cauchy approximation when they are clear from
      context. Otherwise, we explicitly write $f_{\limit}(x, \varphi, a, \psi)$,
      for witnesses $\varphi, \psi$.
    \item For all $u, v : \RR$ such that $u \close{\varepsilon} v$ for every
      positive rational $\varepsilon$, and all $a : A(u)$ and $b : A(v)$
      such that $(u, a) \dclose{\varepsilon} (v, b)$ for every positive rational
      $\varepsilon$, a dependent path\footnote{As in
        \cite[\S6.2]{univalentfoundationsprogram2013}, if $u : B(x)$,
        $v : B(y)$, and $p : x = y$, the type of dependent paths from $u$ to $v$
        over $p$ can be expressed $u =_{p}^{B} v := \transport^{B}(p, u) = v$}
      $a =_{\cpath(u, v)}^{A} b$.
    \item For all $q, r : \QQ$ and all $\varepsilon : \Qp$, a proof that if
      $- \varepsilon < q - r < \varepsilon$ then
      \[
        (\rational(q), f_{\rational}(q)) \dclose{\varepsilon} (\rational(r), f_{\rational}(r)).
      \]
    \item For all rationals $q : \QQ$, all $\varepsilon, \delta : \Qp$ with
      $\varepsilon - \delta > 0$, all Cauchy approximations $y : \Qp \to \RR$,
      and all dependent Cauchy approximations $b$ over $y$, a proof that if
      $\rational(q) \close{\varepsilon - \delta} y_{\delta}$ and
      $(\rational(q), f_{\rational}(q)) \dclose{\varepsilon - \delta}
      (y_{\delta}, b_{\delta})$ then
      \[
        (\rational(q), f_{\rational}(q)) \dclose{\varepsilon} (\limit(y), f_{\limit}(y, b)).
      \]
    \item For all Cauchy approximations $x : \Qp \to \RR$, all dependent Cauchy
      approximations $a$ over $x$, all $r : \QQ$, and all $\varepsilon, \delta$
      with $\varepsilon - \delta > 0$, a proof that if
      $x_{\delta} \close{\varepsilon - \delta} \rational(r)$ and
      $(x_{\delta}, a_{\delta}) \dclose{\varepsilon - \delta} (\rational(r),
      f_{\rational}(r))$ then
      \[
        (\limit(x), f_{\limit}(x, a)) \dclose{\varepsilon} (\rational(r),
        f_{\rational}(r)).
      \]
    \item For all Cauchy approximations $x, y : \Qp \to \RR$, all dependent
      Cauchy approximations $a$ and $b$ over $x$ and $y$, respectively, and all
      $\varepsilon, \delta, \eta : \Qp$ with $\varepsilon - (\delta + \eta) > 0$,
      a proof that if $x_{\delta} \close{\varepsilon - (\delta + \eta)} y_{\eta}$ and
      $(x_{\delta}, a_{\delta}) \dclose{\varepsilon - (\delta + \eta)}
      (y_{\eta}, b_{\eta})$ then
      \[
        (\limit(x), f_{\limit}(x, a)) \dclose{\varepsilon} (\limit(y),
        f_{\limit}(y, b)).
      \]
    \item For all $u, v : \RR$, all $\varepsilon : \Qp$, all
      $a : A(u)$ and $b : A(v)$, and all $\varphi : u \close{\varepsilon} v$,
      a proof that the type
      $(u, a) \dclose{\varepsilon}[\varphi] (v, b)$ is a proposition.
  \end{itemize}
  Under these hypotheses, we obtain functions
  \begin{align*}
    \induction_{\RR} & : \prd{u : \RR} A(u), \\
    \induction_{\closesym} & : \prd{u, v : \RR} \prd{\varepsilon : \Qp} \prd{\varphi : u \close{\varepsilon} v} (u, \induction_{\RR}(u)) \dclose{\varepsilon}[\varphi] (v, \induction_{\RR}(v))
  \end{align*}
  which satisfy the following computation rules
  \begin{align*}
    \induction_{\RR}(\rational(q)) & \doteq f_{\rational}(q), \\
    \induction_{\RR}(\limit(x, \varphi)) & \doteq f_{\limit}(x, \varphi, \varepsilon \mapsto \induction_{\RR}(x_{\varepsilon}), \psi),
  \end{align*}
  where $q : \QQ$, $x : \Qp \to \RR$, $\varphi$ witnesses that $x$ is a Cauchy
  approximation, and for $\varepsilon, \delta : \Qp$, we define
  $\psi(\varepsilon, \delta) := \induction_{\closesym}(x_{\varepsilon},
  x_{\delta}, \varepsilon + \delta, \varphi(\varepsilon, \delta))$. Thus $\psi$
  witnesses that the induced function
  $\varepsilon \mapsto \induction_{\RR}(x_{\varepsilon})$ is a dependent Cauchy
  approximation over $x$.
\end{definition}

We obtain useful special cases of the general induction principle by
trivializing one of two codomain families. If $\dclosesym$ is the constant
family returning the unit type, we obtain an induction principle for $\RR$
alone. If $A$ is constant at the unit type, we obtain an induction principle for
the closeness relation $\closesym$.

\begin{definition}[{\textcite[\S~11.3.2]{univalentfoundationsprogram2013}}]
  \label{def:rr-induction}
  By taking the family $\dclosesym$ in \Cref{def:general-induction} to be
  constant at $\unittype$, we obtain a special case of the general induction
  principle referred to as \textbf{$\RR$-induction}. Thus, to construct a
  section of a type family $A : \RR \to \mathcal{U}$, it suffices to specify:
  \begin{itemize}
    \item For each $q : \QQ$, an element $f_{\rational}(q) : A(\rational(q))$.
    \item For each Cauchy approximation $x : \Qp \to \RR$ and each dependent
      function $a : \prd{\varepsilon : \Qp} A(x_{\varepsilon})$, an element
      $f_{\limit}(x, a) : A(\limit(x))$.
    \item For all $u, v : \RR$ such that $u \close{\varepsilon} v$ for all
      $\varepsilon : \Qp$, and all $a : A(u)$ and $b : A(v)$, a dependent path
      $a =^{A}_{\cpath(u, v)} b$.
  \end{itemize}
  Under these hypotheses, $\RR$-induction yields a function
  \[
    \induction^{\RR} : \prd{u : \RR} A(u)
  \]
  which satisfies the same computation rules as $(\RR,
  \closesym)$-induction. Since $\dclosesym$ is constant at $\unittype$, proofs
  for the dependent Cauchy approximation condition are trivially satisfied.
\end{definition}

\begin{definition}[{\textcite[\S~11.3.2]{univalentfoundationsprogram2013}}]
  \label{def:close-induction}
  By taking the family $A$ in \Cref{def:general-induction} to be constant at
  $\unittype$, we obtain a special case of the general induction principle
  referred to as \textbf{$\closesym$-induction}. Thus, to construct a section of
  a type family
  \[
    \dclosesym : \prd{u, v : \RR} \prd{\varepsilon : \Qp} u
    \close{\varepsilon} v \to \mathcal{U}
  \]
  it suffices to specify data corresponding to the constructors of $\closesym$,
  where $q, r$ are rationals, $\varepsilon, \delta, \eta$ are positive
  rationals, and $x, y$ are Cauchy approximations:
  \begin{itemize}
  \item If $- \varepsilon < q - r < \varepsilon$ then
    $\rational(q) \dclose{\varepsilon} \rational(r)$.
  \item If $\varepsilon - \delta > 0$ and
    $\rational(q) \close{\varepsilon - \delta} y_{\delta}$ then
    $\rational(q) \dclose{\varepsilon} \limit(y)$.
  \item If $\varepsilon - \delta > 0$ and
    $x_{\delta} \close{\varepsilon - \delta} \rational(r)$ then
    $\limit(x) \dclose{\varepsilon} \rational(r)$.
  \item If $\varepsilon - (\delta + \eta) > 0$,
    $x_{\delta} \close{\varepsilon - (\delta + \eta)} y_{\eta}$, and
    $x_{\delta} \dclose{\varepsilon - (\delta + \eta)}$ then
    $\limit(x) \dclose{\varepsilon} \limit(y)$.
  \item For all $u, v : \RR$, all $\varepsilon : \Qp$, and all
    $\varphi : u \close{\varepsilon} v$, the type
    $u \dclose{\varepsilon}[\varphi] v$ is a proposition.
  \end{itemize}
  Under these hypotheses, $\closesym$-induction yields a function
  \[
    \induction^{\closesym} :
    \prd{u, v : \RR} \prd{\varepsilon : \Qp}
    u \close{\varepsilon} v \to
    u \dclose{\varepsilon} v.
  \]
  Note, since $A$ is constant at $\unittype$ and $\closesym$ takes values in
  propositions, there is no meaningful dependence on the endpoints in the
  codomain or the proof of closeness in the domain.
\end{definition}

The principle of recursion obtained by making the general induction principle
non-dependent is not as useful as we would like. This recursion principle, which
is referred to as ``ordinary'' $(\RR, \closesym)$-recursion in
\cite{univalentfoundationsprogram2013}, states that to construct a function
$f : \RR \to A$, it suffices to provide:
\begin{itemize}
\item For every $q : \QQ$, an element $f(\rational(q)) : A$.
\item For every Cauchy approximation $x : \Qp \to \RR$, an element
  $f(\limit(x)) : A$, assuming $f$ has been defined on $x_{\varepsilon}$ for all
  rational $\varepsilon > 0$.
\item For every $u, v : \RR$ such that $u \close{\varepsilon} v$ for all
  rational $\varepsilon > 0$, a proof that $f(u) = f(v)$.
\end{itemize}
As explained in \cite{univalentfoundationsprogram2013}, the last condition is
generally difficult to prove unless we have extra information specifying how $f$
behaves on $\varepsilon$-close reals. In other words, we need a way to measure
approximate closeness between the images of $\varepsilon$-close reals in the
codomain. This is exactly the role played by the family $\dclosesym$ in the
general induction principle, but since $A$ is now non-dependent, $\dclosesym$ is
simply a type family of the form $A \to A \to \Qp \to \mathcal{U}$. Accordingly,
\cite{univalentfoundationsprogram2013} introduces a recursion principle whose
hypotheses include exactly this extra information.

\begin{definition}[{\textcite[\S~11.3.2]{univalentfoundationsprogram2013}}]
  \label{def:enhanced-recursion}
  Let $A$ be a type and let $a \dclose{\varepsilon} b$ be a type family indexed
  by $a, b : A$ and $\varepsilon : \Qp$. The \textbf{enhanced principle of
    $(\RR, \closesym)$-recursion}, which we will refer to simply as
  \textbf{$(\RR, \closesym)$-recursion} since it is the only recursion principle
  used here, states that to construct a function $\RR \to A$, it suffices to
  provide the following data:
  \begin{itemize}
  \item A function $f_{\rational} : \QQ \to A$.
  \item For every Cauchy approximation $x : \Qp \to \RR$ and every map
    $f' : \Qp \to A$ such that
    \[
      f'(\varepsilon) \dclose{\varepsilon + \delta} f'(\delta)
    \]
    for all rational $\varepsilon, \delta > 0$, an element
    $f_{\limit}(x, f') : A$. Note that this condition amounts to asserting that
    $f'$ is a dependent Cauchy approximation over $x$, albeit with the
    dependence removed, since $A$ is no longer a type family over $\RR$. In
    \cite{univalentfoundationsprogram2013}, this is referred to as ``a Cauchy
    approximation with respect to $\dclosesym$''. As with the general induction
    principle, we suppress the relevant witnesses when they are clear from
    context; otherwise we write $f_{\limit}(x, \varphi, f', \psi)$ for witnesses
    $\varphi$ and $\psi$.
  \item For all $a, b : A$, if $a \dclose{\varepsilon} b$ for all
    $\varepsilon : \Qp$, then $a = b$. This is referred to as the
    \emph{separatedness} condition for $\dclosesym$.
  \item For all $a, b : A$ and all $\varepsilon : \Qp$, the type
    $a \dclose{\varepsilon} b$ is a proposition; that is, the family
    $\dclosesym$ is a mere relation.
  \item For all $q, r : \QQ$ and all $\varepsilon : \Qp$, if
    $- \varepsilon < q - r < \varepsilon$ then
    $f_{\rational}(q) \dclose{\varepsilon} f_{\rational}(r)$.
  \item For all $q : \QQ$, all $\varepsilon, \delta : \Qp$ with
    $\varepsilon - \delta > 0$, all Cauchy approximations $y : \Qp \to \RR$, and
    all Cauchy approximations $g' : \Qp \to A$ with respect to $\dclosesym$, if
    $\rational(q) \close{\varepsilon - \delta} y_{\delta}$ and
    $f_{\rational}(q) \dclose{\varepsilon - \delta} g'(\delta)$ then
    \[
      f_{\rational}(q) \dclose{\varepsilon} f_{\limit}(y, g').
    \]
  \item For all Cauchy approximations $x : \Qp \to \RR$, all Cauchy
    approximations $f' : \Qp \to A$ with respect to $\dclosesym$, all $r : \QQ$,
    and all $\varepsilon, \delta : \Qp$ with $\varepsilon - \delta > 0$, if
    $x_{\delta} \close{\varepsilon - \delta} \rational(r)$ and
    $f'(\delta) \dclose{\varepsilon - \delta} f_{\rational}(r)$ then
    \[
      f_{\limit}(x, f') \dclose{\varepsilon} f_{\rational}(r).
    \]
  \item For all Cauchy approximations $x, y : \Qp \to \RR$, all Cauchy
    approximations $f', g' : \Qp \to A$ with respect to $\dclosesym$, and all
    $\varepsilon, \delta, \eta : \Qp$ with $\varepsilon - (\delta + \eta) > 0$,
    if $x_{\delta} \close{\varepsilon - (\delta + \eta)} y_{\eta}$ and
    $f'(\delta) \dclose{\varepsilon - (\delta + \eta)} g'(\eta)$ then
    \[
      f_{\limit}(x, f') \dclose{\varepsilon} f_{\limit}(y, g').
    \]
  \end{itemize}
  Under these hypotheses, $(\RR, \closesym)$-recursion yields functions
  \begin{align*}
    \recursion_{\RR} & : \RR \to A, \\
    \recursion_{\closesym} & : \prd{u, v : \RR} \prd{\varepsilon : \Qp} u \close{\varepsilon} v \to \recursion_{\RR}(u) \dclose{\varepsilon} \recursion_{\RR}(v)
  \end{align*}
  such that for all $q : \QQ$ and all $x : \Qp \to \RR$ equipped with a witness
  $\varphi$ that $x$ is a Cauchy approximation, the computation rules
  \begin{align*}
    \recursion_{\RR}(\rational(q)) & = f_{\rational}(q), \\
    \recursion_{\RR}(\limit(x, \varphi)) & = f_{\limit}(x, \varphi, \varepsilon \mapsto \recursion_{\RR}(x_{\varepsilon}), \psi)
  \end{align*}
  are satisfied, where for all $\varepsilon, \delta : \Qp$,
  \[
    \psi(\varepsilon, \delta) := \recursion_{\closesym}(x_{\varepsilon},
    x_{\delta}, \varepsilon + \delta, \varphi(\varepsilon, \delta))
  \]
  witnesses that the map
  $\varepsilon \mapsto \recursion_{\RR}(x_{\varepsilon})$ is a Cauchy
  approximation with respect to $\dclosesym$.
\end{definition}

\section{Closeness}
\label{section:closeness}

The induction and recursion principles allow us to build up a collection of
basic properties about the closeness relation. In turn, these properties make
closeness the main tool for constructing functions on the reals, and ultimately,
for equipping $\RR$ with its algebraic and order-theoretic structure. For
example, $\RR$-induction shows that $\varepsilon$-closeness is reflexive.

\begin{lemma}[{\textcite[Lemma 11.3.8]{univalentfoundationsprogram2013}}]
  \label{lem:close-reflexive}
  For all rational $\varepsilon > 0$, we have $u \close{\varepsilon} u$. In
  other words, for each rational $\varepsilon > 0$, the binary relation
  $\close{\varepsilon}$ is reflexive.
\end{lemma}

Combining the previous lemma with the fact that the family of closeness
relations is separated and takes values in propositions, we can apply Theorem
7.2.2 of \cite{univalentfoundationsprogram2013}, which states that a type
equipped with a reflexive mere relation implying identity is a set.

\begin{corollary}[{\textcite[Theorem 11.3.9]{univalentfoundationsprogram2013}}]
  \label{corollary:reals-set}
  The HoTT book reals form a set.
\end{corollary}

Similarly, $\closesym$-induction shows that $\varepsilon$-closeness is symmetric.

\begin{lemma}[{\textcite[Lemma 11.3.12]{univalentfoundationsprogram2013}}]
  \label{lem:close-symmetric}
  For all rational $\varepsilon > 0$, if $u \close{\varepsilon} v$ then
  $v \close{\varepsilon} u$, that is, the relation $\close{\varepsilon}$ is
  symmetric.
\end{lemma}

We also need closeness to satisfy analogues of familiar metric properties. The
first is a version of the triangle inequality formulated for relations indexed
by a positive rational.

\begin{definition}[{\textcite[\S~11.3.2]{univalentfoundationsprogram2013}}]
  \label{definition:triangle-inequality}
  A relation 
  \[
    \approx : \RR \to \RR \to \Qp \to \mathcal{U}
  \]
  satisfies the \textbf{triangle inequality} if, for all
  $u, v, w : \RR$ and all $\varepsilon, \delta : \Qp$, if
  $u \approx_{\varepsilon} v$ and $v \approx_{\delta} w$ then
  $u \approx_{\varepsilon + \delta} w$.
\end{definition}

The second condition we need is roundedness.

\begin{definition}[{\textcite[\S~11.3.2]{univalentfoundationsprogram2013}}]
  \label{def:rounded}
  A relation
  \[
    \approx : \RR \to \RR \to \Qp \to \mathcal{U}
  \]
  is \textbf{rounded} if, for all $u, v : \RR$ and all $\varepsilon : \Qp$,
  \[
    u \approx_{\varepsilon} v \iff
    \exists \theta : \Qp, (\theta < \varepsilon) \times (u \approx_{\varepsilon - \theta} v).
  \]
  The implication from left to right is called \textbf{openness}, and the
  converse implication is called \textbf{monotonicity}.
\end{definition}

At first this openness condition may seem counterintuitive: why should
$\varepsilon$-closeness imply $(\varepsilon - \theta)$-closeness for some
strictly positive $\theta$? The point is that $u \close{\varepsilon} v$ is
intended to express a \emph{strict} inequality. Indeed, once the algebraic and
order structure of $\RR$ has been developed, Theorem~11.3.44 of
\cite{univalentfoundationsprogram2013} shows
\[
  \left(u \close{\varepsilon} v\right) \simeq
  \left(\lvert u - v \rvert < \rational(\varepsilon)\right).
\]
From this perspective, the openness condition becomes much more natural.

This phenomenon is visible in the rational case. We can show that the relation
$\sim : \QQ \to \QQ \to \Qp \to \mathcal{U}$ given by
\[
  q \close{\varepsilon} r := \lvert q - r \rvert < \varepsilon
\]
satisfies the openness condition.

\begin{lemma}
  The rational closeness relation is open.
\end{lemma}
\begin{proof}
  Suppose $\lvert q - r \rvert < \varepsilon$. Define
  \[
    \theta := \frac{\varepsilon - \lvert q - r \rvert}{2}.
  \]
  Then $\varepsilon - \lvert q - r \rvert > 0$, so $0 < \theta$. Moreover,
  \[
    \varepsilon - \theta = \varepsilon - \frac{\varepsilon - \lvert q - r \rvert}{2} = \frac{\lvert q - r \rvert + \varepsilon}{2}
  \]
  which is the midpoint between $\lvert q - r \rvert$ and $\varepsilon$, and
  since $\lvert q - r \rvert < \varepsilon$, it lies strictly between
  $\lvert q - r \rvert$ and $\varepsilon$. In particular,
  \[
    \lvert q - r \rvert < \varepsilon - \theta.
  \]
  Hence $q \close{\varepsilon - \theta} r$.
\end{proof}

However, reasoning with midpoints in this manner assumes that closeness is
induced by a distance metric. At this stage, the closeness relation has only
been given inductively, so we do not yet have the characterization of
Theorem~11.3.44 available. Since the characterization itself depends on the
triangle inequality and roundedness, appealing to it here would be circular. We
therefore need a different way to analyze the inductively defined closeness
relation.

There is a second difficulty. \Cref{def:hott-reals} defines the closeness
relation by constructors, so its clauses tell us how to \emph{construct}
witnesses of $\varepsilon$-closeness. For example, the rational-rational
constructor gives a map
\[
    -\varepsilon < q - r < \varepsilon \to
    \rational(q) \close{\varepsilon} \rational(r).
\]
Thus the displayed inequality is sufficient to produce a proof that
$\rational(q)$ and $\rational(r)$ are $\varepsilon$-close.

What is missing is the converse implication. If we are instead given a proof of
$\rational(q) \close{\varepsilon} \rational(r)$ as a hypothesis, the inductive
definition does not by itself allow us to conclude that
$- \varepsilon < q - r < \varepsilon$. The same issue arises in the other three
cases involving limits. In other words, the constructor clauses for closeness do
not yet provide a case-by-case characterization of closeness when it appears as
an assumption.

The strategy in \cite{univalentfoundationsprogram2013} is therefore to define an
auxiliary closeness relation $\approx$ by recursion, so that when its arguments
are rational points or limits, the expression $u \approx_{\varepsilon} v$
unfolds according to one of four explicit clauses. In this sense, $\approx$
computes on the point constructors of the reals. This gives a closeness relation
whose behavior can be read off immediately from the constructors used in its
arguments.

\begin{theorem}[{\textcite[Theorem 11.3.16]{univalentfoundationsprogram2013}}]
  \label{theorem:alternative-closeness}
  There is a family of mere relations $\approx : \RR \to \RR \to \Qp \to \mathcal{U}$ such that
  \begin{align*}
    (\rational(q) \approx_{\varepsilon} \rational(r)) & := (- \varepsilon < q - r < \varepsilon) \\
    (\rational(q) \approx_{\varepsilon} \limit(y)) & := \exists \delta : \Qp, \rational(q) \approx_{\varepsilon - \delta} y_{\delta} \\
    (\limit(x) \approx_{\varepsilon} \rational(r)) & :=  \exists \delta : \Qp, x_{\delta} \approx_{\varepsilon - \delta} \rational(r) \\
    (\limit(x) \approx_{\varepsilon} \limit(y)) & := \exists \delta, \eta : \Qp, x_{\delta} \approx_{\varepsilon - (\delta + \eta)} y_{\eta}.
  \end{align*}
  Moreover, $\approx$ is rounded and satisfies the mixed triangle laws
  \begin{align*}
    (u \approx_{\varepsilon} v) \to (v \close{\delta} w) \to (u \approx_{\varepsilon + \delta} w), \\
    (u \close{\varepsilon} v) \to (v \approx_{\delta} w) \to (u \approx_{\varepsilon + \delta} w).
  \end{align*}
\end{theorem}

The advantage of $\approx$ is that these explicit clauses provide the converse
information that was missing from $\closesym$. To use this information for the
original closeness relation, however, we need to show that the two
relations coincide.

\begin{theorem}[{\textcite[Theorem 11.3.32]{univalentfoundationsprogram2013}}]
  \label{theorem:close-eq-close'}
  For any $u, v : \RR$ and $\varepsilon : \Qp$, we have
  \[
    u \close{\varepsilon} v \iff u \approx_{\varepsilon} v.
  \]
  Since both sides are mere propositions, it follows that
  \[
    (u \close{\varepsilon} v) = (u \approx_{\varepsilon} v).
  \]
\end{theorem}

Once this equivalence is established, the explicit characterization of $\approx$
can be transferred back to $\closesym$, along with the roundedness and triangle
properties built into the construction of $\approx$.

\begin{corollary}[{\textcite[Corollary 11.3.33]{univalentfoundationsprogram2013}}]
  Closeness is rounded and satisfies the triangle inequality.
\end{corollary}

The results of this section supply the basic structural properties of the
closeness relation needed in the remainder of the development. In particular,
roundedness and the triangle inequality make closeness workable as a substitute
for metric reasoning, while the auxiliary relation $\approx$ provides the
explicit case analysis needed to extract usable information from closeness
hypotheses. We next use these properties to provide lifts of Lipschitz and
non-expanding maps from the rationals to the reals.

\section{Lifting Lipschitz and Nonexpanding Maps}
\label{section:lifting}

One conceptual advantage of the HoTT book reals is that their construction
forces us to make explicit the principle by which maps on the rationals extend
to the reals. In the classical quotient construction of the Cauchy reals, we
typically define operations on representatives and then prove that they are
well-defined on equivalence classes. This procedure relies on the fact that a
suitably continuous map on the dense subspace $\QQ$ extends to the completion
$\RR$, but the corresponding well-definedness argument can easily recede into
the background as a routine technical step. In the higher-inductive setting,
however, there are no such representatives, so we must formulate and prove
extension principles explicitly using the recursion principle for the reals.

The HoTT book develops extension lemmas for Lipschitz and non-expanding maps
\cite{univalentfoundationsprogram2013}. These hypotheses are strong enough to
lift important operations such as addition and $\min$/$\max$, but they do not
capture every operation we ultimately want to define on the reals. In
particular, neither multiplication nor reciprocal are globally Lipschitz, so
their definitions cannot be obtained from the following lemmas alone and require
additional work in \Cref{section:algebra-order}. Accordingly, the results of
this section should be read as a collection of extension principles tailored to
the later algebraic development, rather than as a final characterization of the
conditions under which maps can be lifted from $\QQ$ to $\RR$. A broader
extension theorem, for example for uniformly continuous maps, would be desirable
but is beyond the scope of this thesis; see \Cref{chapter:future-work}.

The first extension principle concerns Lipschitz maps, whose defining condition
is naturally expressed in terms of the closeness relation.

\begin{definition}[{\textcite[Definition 11.3.14]{univalentfoundationsprogram2013}}]
  \label{definition:lipschitz}
  In the rational case, a map $f : \QQ \to \RR$ is \textbf{Lipschitz} if it
  comes equipped with an element $L : \Qp$, the \textbf{Lipschitz} constant,
  such that, for all $q, r : \QQ$ and $\varepsilon : \Qp$,
  \[
    \lvert q - r \rvert < \varepsilon \implies f(q) \close{L \varepsilon} f(r).
  \]
  Analogously, for the real case, a map $g : \RR \to \RR$ is \textbf{Lipschitz}
  if it comes equipped with an element $L : \Qp$ such that, for all $u, v : \RR$
  and $\varepsilon : \Qp$,
  \[
    u \close{\varepsilon} v \implies g(u) \close{L \varepsilon} g(v).
  \]
\end{definition}

With this definition in place, we can state the extension principle.

\begin{lemma}[{\textcite[Lemma 11.3.15]{univalentfoundationsprogram2013}}]
  \label{lemma:lift-lipschitz}
  Let $f : \QQ \to \RR$ be Lipschitz with constant $L : \Qp$. Then there is a
  map $\overline{f} : \RR \to \RR$ which is Lipschitz with the same constant $L$
  and satisfies
  \[
    \overline{f}(\rational(q)) = f(q)
  \]
  for all $q : \QQ$.
\end{lemma}

For the binary operations needed later, \cite{univalentfoundationsprogram2013}
does not state a fully general two-variable Lipschitz extension lemma. Instead,
it isolates the special case of maps that are non-expanding, that is, Lipschitz
with constant $L \le 1$ in each variable separately. This hypothesis is more
restrictive than allowing an arbitrary Lipschitz constant, but apart from
multiplication and reciprocal it covers the operations we want to extend in
\Cref{section:algebra-order}.

\begin{definition}[{\textcite[Lemma 11.3.40] {univalentfoundationsprogram2013}}]
  A map $f : \QQ \to \QQ$ is \textbf{non-expanding} if, for all $q, r : \QQ$,
  \[
    \lvert f(q) - f(r) \rvert \le \lvert q - r \rvert.
  \]
  Equivalently, if $\lvert q-r\rvert < \varepsilon$, then
  $\lvert f(q)-f(r)\rvert < \varepsilon$ for every $\varepsilon : \Qp$.  This is
  the form that generalizes directly to the reals: a map $f : \RR \to \RR$ is
  non-expanding if, for all $u, v : \RR$ and $\varepsilon : \Qp$,
  \[
    u \close{\varepsilon} v \implies f(u) \close{\varepsilon} f(v).
  \]
  Moreover, a function $\QQ \to \QQ \to \QQ$ is non-expanding if, for all
  $r : \QQ$, the map $f(-, r)$ is non-expanding and, for all $q : \QQ$, the map
  $f(q, -)$ is non-expanding. The two-variable real case is defined analogously.
\end{definition}

\begin{lemma}[{\textcite[Lemma 11.3.40]{univalentfoundationsprogram2013}}]
  \label{lemma:lift-nonexpanding}
  Let $f : \QQ \to \QQ \to \QQ$ be a non-expanding map. Then there is a non-expanding map $\overline{f} : \RR \to \RR \to \RR$ such that
  \[
    \overline{f}(\rational(q), \rational(r)) = \rational(f(q, r)).
  \]
  for all $q, r : \RR$.
\end{lemma}

In addition to proving that such extensions exist, we also need to know when
they are uniquely determined. The relevant uniqueness statement can be more
generally formulated in terms of arbitrary continuous functions, so we first
express the standard definition of continuity in terms of the closeness
relation.

\begin{definition}[{\textcite[Lemma 11.3.39]{univalentfoundationsprogram2013}}]
  A map $f : \RR \to \RR$ is \textbf{continuous at a point} $u : \RR$ if for all
  $\varepsilon : \Qp$, there merely exists a $\delta : \Qp$ such that for all
  $v : \RR$,
  \[
    u \close{\delta} v \implies f(u) \close{\varepsilon} f(v).
  \]
  A map $f : \RR \to \RR$ is \textbf{continuous} if it is continuous for every
  point $u : \RR$.
\end{definition}

This definition encompasses the extensions constructed above. Indeed, any
Lipschitz map is continuous: given $\varepsilon$, choose
$\delta := \varepsilon / L$; then $u \close{\delta} v$ implies
$f(u) \close{\varepsilon} f(v)$. As the special case of a Lipschitz map with
constant $L \le 1$, every univariate non-expanding map is continuous. Likewise,
if $f : \RR \to \RR \to \RR$ is non-expanding in each variable separately, then
it will be continuous in each variable separately.

\begin{lemma}[{\textcite[Lemma 11.3.39]{univalentfoundationsprogram2013}}]
  \label{lemma:continuous-extension-unique}
  Let $f, g : \RR \to \RR$ be continuous. If
  $f \circ \rational = g \circ \rational$ then $f = g$.
\end{lemma}

The previous lemma gives us a way to extend identities from the rationals to
identities on the reals. For example, once addition on $\RR$ has been
constructed, the maps
\[
  u \mapsto u + 0 \qquad\text{and}\qquad u \mapsto u
\]
are continuous and agree on rational inputs, since $q + 0 = q$ for all
$q : \QQ$. The previous lemma therefore yields the right unit law for real
addition:
\[
  u + 0 = u
\]
for all $u : \RR$.

The authors of \cite{univalentfoundationsprogram2013} indicate that identities
in several variables can be extended in the same manner. However, the precise
hypotheses needed to make this work are not completely obvious. In the binary
case, for example, to prove commutativity of addition, should we require
addition to be jointly continuous as a function
\[
  + : \RR \times \RR \to \RR,
\]
or does continuity in each variable separately suffice? At first glance, the
latter condition may seem too weak. In fact, separate continuity is sufficient,
because the one-variable uniqueness lemma can be applied one coordinate at a
time. Since this result is used repeatedly later and is not stated explicitly in
\cite{univalentfoundationsprogram2013}, we formulate and prove it here.

\begin{lemma}
  \label{lemma:continuous-extension-unique2}
  Let $f, g : \RR \to \RR \to \RR$ be maps that are continuous in both variables
  separately. That is, for every $u : \RR$, the maps
  \[
    f(u, -), g(u, -) : \RR \to \RR
  \]
  are continuous, and for every $v : \RR$, the maps
  \[
    f(-, v), g(-, v) : \RR \to \RR
  \]
  are continuous.
  Suppose that 
  \[
    f(\rational(q), \rational(r)) = g(\rational(q), \rational(r))
  \]
  for all $q, r : \QQ$. Then
  \[
    f(u, v) = g(u, v)
  \]
  for all $u, v : \RR$, and hence $f = g$.
\end{lemma}
\begin{proof}
  Fix $u : \RR$. We show that
  \[
    f(u, v) = g(u, v)
  \]
  for all $v : \RR$.

  By the assumed continuity of $f(u, -)$ and $g(u, -)$, the maps
  \[
    v \mapsto f(u, v) \qquad\text{and}\qquad v \mapsto g(u, v)
  \]
  are continuous. Thus, by \Cref{lemma:continuous-extension-unique}, it suffices
  to prove that
  \[
    f(u, \rational(r)) = g(u, \rational(r))
  \]
  for all $r : \QQ$.

  Fix $r : \QQ$. By the assumed continuity of $f(-, \rational(r))$ and
  $g(-, \rational(r))$, the maps
  \[
    w \mapsto f(w, \rational(r)) \qquad\text{and}\qquad
    w \mapsto g(w, \rational(r))
  \]
  are continuous. Moreover, for every $q : \QQ$, the hypothesis gives
  \[
    f(\rational(q), \rational(r)) = g(\rational(q), \rational(r)).
  \]
  Applying \Cref{lemma:continuous-extension-unique}, we conclude that
  \[
    f(w, \rational(r)) = g(w, \rational(r))
  \]
  for all $w : \RR$. In particular, evaluating at $w = u$ yields
  \[
    f(u, \rational(r)) = g(u, \rational(r)).
  \]

  Since $r : \QQ$ was arbitrary, the required agreement on rational inputs
  follows, and therefore
  \[
    f(u, v) = g(u, v)
  \]
  for all $v : \RR$. Since $u : \RR$ was arbitrary, it follows that
  \[
    f(u, v) = g(u, v)
  \]
  for all $u, v : \RR$, and hence $f = g$.
\end{proof}

The previous lemma can be packaged as a convenient extension law for binary
operations. This is the form that used repeatedly in
\Cref{section:algebra-order} to extend identities from the rationals to the
reals.

\begin{corollary}
  \label{corollary:continuous-extension-law}
  Let $f, g : \RR \to \RR \to \RR$ and $f', g' : \QQ \to \QQ \to \QQ$. Suppose
  that
  \[
    f(\rational(q), \rational(r)) = \rational(f'(q, r))
  \]
  and
  \[
    g(\rational(q), \rational(r)) = \rational(g'(q, r))
  \]
  for all $q, r : \QQ$, and that
  \[
    f'(q, r) = g'(q, r)
  \]
  for all $q, r : \QQ$. Assume moreover that $f$ and $g$ are continuous in each
  variable separately. Then
  \[
    f(u, v) = g(u, v)
  \]
  for all $u, v : \RR$, and hence $f = g$.
\end{corollary}
\begin{proof}
  For any $q, r : \QQ$, the assumptions give a path
  \[
    f(\rational(q), \rational(r)) =
    \rational(f'(q, r)) =
    \rational(g'(q, r)) =
    g(\rational(q), \rational(r)).
  \]
  Thus $f$ and $g$ agree on rational pairs. Since both maps are
  continuous in both variables separately, the result follows by
  \Cref{lemma:continuous-extension-unique2}.
\end{proof}

The same pattern applies to functions of three or more variables: we extend
agreement one coordinate at a time, using separate continuity in the relevant
variable at each step. We do not formulate a general $n$-ary version here, since
in \Cref{section:algebra-order} we only extend identities involving at most
three variables, and a fully general statement would require additional
bookkeeping for uncurried $n$-ary functions.

This completes the extension principles needed in the next section. The lifting
lemmas provide operations on $\RR$, while the uniqueness lemmas explain
precisely how identities on $\QQ$ induce identities on $\RR$.

\section{Algebra and Order}
\label{section:algebra-order}

The previous three sections provide the tools needed to equip the HoTT book
reals with their algebraic and order-theoretic structure. Because the
corresponding Agda development for algebra and order spans more than six
thousand lines, we do not reproduce it result by result. Instead, we organize
this section around three case studies showing how the tools from the previous
sections are used in practice. First, we lift negation, addition, minimum, and
maximum from the rationals and illustrate how algebraic identifiers transfer
from $\QQ$ to $\RR$; this also yields the induced non-strict order. Second, we
define strict order and develop the lemmas needed to show that addition
preserves and reflects it. Third, we construct multiplication and reciprocal,
which require extra work beyond the extension principles of
\Cref{section:lifting}.

This organization also reflects the division of labor among our sources. For the
first study, we follow the HoTT book closely. For the development of strict
inequality, and especially for the proof that addition preserves and reflects
strict inequality, we follow Gilbert \cite{gilbert2017}, since the HoTT book
does not provide that argument. For multiplication and reciprocal, we likewise
depart from the HoTT book's squaring-based construction and instead adopt
Gilbert's more direct construction using iterated Lipschitz extension. The
section closes by assembling the results into the statement that the HoTT book
reals form an Archimedean ordered field.

We begin with the first case study, concerning the operations obtained directly
by lifting from the rationals. Negation is obtained by Lipschitz extension,
while addition and $\min$/$\max$ are obtained by non-expanding extension. These
constructions all follow a common pattern. First, we prove that the rational
version of the operation is Lipschitz or non-expanding. We then use this result
to extend the operation to the reals by applying
\Cref{lemma:lift-lipschitz} or \Cref{lemma:lift-nonexpanding}. In this way we
obtain operations
\[
  - : \RR \to \RR,\qquad + : \RR \to \RR \to \RR,\qquad \min, \max : \RR \to \RR \to \RR.
\]
Then, for each basic algebraic law the operation should satisfy, we use or
prove the corresponding result on the rationals and apply
\Cref{corollary:continuous-extension-law} and its analogues for functions of
higher arity.

Here, we walk through the construction of the
$\max$ operation and prove that it satisfies
\begin{equation}
  \label{eq:max-add-left}
  \max(a + u, a + v) = a + \max(u, v)
\end{equation}
for all $a, u, v : \RR$. Once the order relation has been defined, this identity
will be the key to showing that addition preserves and reflects order, that is,
that we have
\[
  u \le v \iff a + u \le a + v
\]
for all $a, u, v : \RR$. In addition to illustrating the common pattern well, this
example is also not covered explicitly in either
\cite{univalentfoundationsprogram2013} or \cite{gilbert2017}.

We begin by showing that the $\max$ operation on rationals is non-expanding in
both arguments separately. To do so, we first record a convenient closed-form
expression for $\max(q,r)$ in terms of addition, subtraction, and absolute
value. This identity will be used twice: first to establish the required
non-expandingness of rational $\max$, and then to prove its translation
invariance on the rationals, which will later be lifted to the reals.

\begin{lemma}
  \label{lemma:max-midpoint-half-distance}
  For all rational $q, r : \QQ$, we have
  \[
    \max(q, r) = \frac{q + r}{2} + \frac{\lvert q - r \rvert}{2}.
  \]
\end{lemma}
\begin{proof}
  This follows by cases on $q \le r$ and $r \le q$.
\end{proof}

We first use this identity to prove that rational $\max$ is non-expanding in
each argument.

\begin{lemma}
  \label{lemma:max-rational-nonexpanding}
  The max operation on the rationals is non-expanding in both variables
  separately.
\end{lemma}
\begin{proof}
  Fix $q, r, s : \QQ$. We show
  \begin{align}
    \left\lvert \max(q, s) - \max(r, s) \right\rvert & \le \left\lvert q - r \right\rvert, \label{eq:max-rational-nonexpanding-first}\\
    \left\lvert \max(q, r) - \max(q, s) \right\rvert & \le \left\lvert r - s \right\rvert. \label{eq:max-rational-nonexpanding-second}
  \end{align}
  It suffices to prove \cref{eq:max-rational-nonexpanding-first}, since
  \cref{eq:max-rational-nonexpanding-second} then follows from the commutativity
  of $\max$ and the symmetry of distance. We have
  \begin{align*}
    \left\lvert \max(q, s) - \max(r, s) \right\rvert
      & = \left\lvert \frac{(q + s) + \lvert q - s \rvert}{2} - \frac{(r + s) + \lvert r - s \rvert}{2} \right\rvert \qquad \text{by \Cref{lemma:max-midpoint-half-distance}}\\
      & = \frac{1}{2} \left\lvert (q - r) + (\lvert q - s \rvert - \lvert r - s \rvert) \right\rvert \\
      & \le \frac{1}{2} \left\lvert q - r \right\rvert + \frac{1}{2} \left\lvert \left\lvert q - s \right\rvert - \left\lvert r - s \right\rvert \right\rvert \qquad \text{by the triangle inequality}\\
      & \le \frac{1}{2} \left\lvert q - r \right\rvert + \frac{1}{2} \left\lvert (q - s) - (r - s) \right\rvert \qquad \text{by the reverse triangle inequality}\\
      & = \frac{1}{2} \left\lvert q - r \right\rvert + \frac{1}{2} \left\lvert q - r \right\rvert \\
      & = \left\lvert q - r \right\rvert,
  \end{align*}
  which proves \cref{eq:max-rational-nonexpanding-first}, and hence the lemma.
\end{proof}

Applying the binary non-expanding extension principle
(\Cref{lemma:lift-nonexpanding}) to the previous lemma, we obtain a function
\[
  \max : \RR \to \RR \to \RR
\]
extending rational $\max$. Since the extension is again non-expanding in each
variable separately, $\max$ is continuous in each variable separately as well.
To illustrate how identities on $\QQ$ are transferred to $\RR$, we now prove
that rational $\max$ is translation invariant and then lift that identity to
the reals.

\begin{lemma}
  \label{lemma:max-translation-invariant-rational}
  For all rational $q, r, a : \QQ$, we have
  \[
    \max(a + q, a + r) = a + \max(q, r).
  \]
\end{lemma}
\begin{proof}
  Using \Cref{lemma:max-midpoint-half-distance}, we compute
  \begin{align*}
    \max(a + q, a + r)
    & = \frac{(a + q) + (a + r)}{2} + \frac{\lvert (a + q) - (a + r) \rvert}{2} \\
    & = a + \frac{q + r}{2} + \frac{\lvert q - r \rvert}{2} \\
    & = a + \max(q, r),
  \end{align*}
  as required.
\end{proof}

To lift \Cref{lemma:max-translation-invariant-rational} to the reals, we must
show that the two sides of the desired identity define separately continuous
maps of three real variables. Most of these continuity claims follow directly
from the fact that addition and $\max$ are non-expanding in each argument
separately. The only slightly nontrivial case is the dependence on the
translation parameter $a$ in expressions such as
\[
  a \mapsto \max(a + u, a + v).
\]
This is a composite of two Lipschitz maps with the binary operation $\max$, and
similar combinations arise repeatedly throughout the formalization. We therefore
record the following binary composition lemma once and use it here as a
representative instance.

\begin{lemma}
  \label{lemma:lipschitz-binary-composition}
  Let $f, g : \RR \to \RR$ and $h : \RR \to \RR \to \RR$. Suppose that
  $f$ is Lipschitz with constant $L$, $g$ is Lipschitz with constant $M$,
  and for every $v : \RR$ the map $u \mapsto h(u,v)$ is Lipschitz with
  constant $N_1$, while for every $u : \RR$ the map $v \mapsto h(u,v)$ is
  Lipschitz with constant $N_2$. Then the composite
  \[
    u \mapsto h(f(u), g(u))
  \]
  is Lipschitz with constant $N_1L + N_2M$.
\end{lemma}
\begin{proof}
  Let $u, v : \RR$, let $\varepsilon : \Qp$, and suppose that
  \[
    u \close{\varepsilon} v.
  \]
  Since $f$ and $g$ are Lipschitz with constants $L$ and $M$ respectively, we
  have
  \[
    f(u) \close{L\varepsilon} f(v) \qquad\text{and}\qquad
    g(u) \close{M\varepsilon} g(v).
  \]

  Fixing the second argument, the Lipschitz hypothesis on $w \mapsto h(w, g(u))$
  gives
  \[
    h(f(u),g(u)) \close{N_1 L \varepsilon} h(f(v),g(u)).
  \]
  Similarly, fixing the first argument, the Lipschitz hypothesis on
  $z \mapsto h(f(v), z)$ yields
  \[
    h(f(v),g(u)) \close{N_2 M \varepsilon} h(f(v),g(v)).
  \]

  Applying the triangle inequality for closeness, we obtain
  \[
    h(f(u),g(u))
    \close{N_1 L \varepsilon + N_2 M \varepsilon}
    h(f(v),g(v)).
  \]
  Since
  \[
    N_1 L \varepsilon + N_2 M \varepsilon
    = (N_1L + N_2M)\varepsilon,
  \]
  this shows that $u \mapsto h(f(u),g(u))$ is Lipschitz with constant
  $N_1L + N_2M$.
\end{proof}

With this auxiliary lemma in hand, we can now transfer translation invariance
of $\max$ from the rationals to the reals.

\begin{lemma}
  \label{lemma:max-translation-invariant-real}
  For all real $a, u, v : \RR$, we have
  \[
    \max(a + u, a + v) = a + \max(u, v).
  \]
\end{lemma}
\begin{proof}
  Define maps $f, g : \RR \to \RR \to \RR \to \RR$ by
  \begin{align*}
    f(a,u,v) & := \max(a + u, a + v), \\
    g(a,u,v) & := a + \max(u,v).
  \end{align*}
  Likewise, define $f', g' : \QQ \to \QQ \to \QQ \to \QQ$ by
  \begin{align*}
    f'(a,q,r) & := \max(a + q, a + r), \\
    g'(a,q,r) & := a + \max(q,r).
  \end{align*}

  By construction of addition and $\max$ on $\RR$ as non-expanding lifts of the
  corresponding rational operations, $f$ and $g$ agree with $f'$ and $g'$ on
  rational inputs. \Cref{lemma:max-translation-invariant-rational} gives
  \[
    f'(a, q, r) = g'(a, q, r)
  \]
  for all $a, q, r : \QQ$. It therefore remains to check that $f$ and $g$ are
  continuous in each variable separately, so that the three-variable extension
  argument described after \Cref{corollary:continuous-extension-law} applies.

  For $f$, the maps
  \begin{align*}
    u & \mapsto \max(a + u, a + v), \\
    v & \mapsto \max(a + u, a + v)
  \end{align*}
  are composites of continuous maps, since addition and $\max$ are non-expanding
  (and hence continuous) in both variables separately. For fixed $u$ and $v$,
  the maps
  \[
    a \mapsto a + u\qquad\text{and}\qquad
    a \mapsto a + v
  \]
  are Lipschitz with constant $1$, while $\max$ is Lipschitz with constant $1$
  in both variables separately. Hence \Cref{lemma:lipschitz-binary-composition}
  shows that
  \[
    a \mapsto \max(a + u, a + v)
  \]
  is Lipschitz with constant $2$, and therefore continuous.

  For $g$, the maps
  \begin{align*}
    u \mapsto a + \max(u, v), \\
    v \mapsto a + \max(u, v)
  \end{align*}
  are again composites of continuous maps. For fixed $u$ and $v$, the map
  \[
    a \mapsto a + \max(u, v)
  \]
  is Lipschitz with constant $1$, since addition on the left is Lipschitz with
  constant $1$. Thus $g$ is also continuous in each variable separately.

  The three-variable analogue of \Cref{corollary:continuous-extension-law}
  therefore yields
  \[
    f(a, u, v) = g(a, u, v)
  \]
  for all $a, u, v : \RR$, which is exactly the desired identity.
\end{proof}

The same transfer argument applies to the remaining algebraic laws for addition,
negation, $\min$, and $\max$. Since the corresponding identities hold on $\QQ$,
they lift to $\RR$ by the uniqueness of continuous extensions. In particular,
addition and negation endow $\RR$ with the structure of an Abelian group, while
$\min$ and $\max$ satisfy the usual associative, commutative, idempotent, and
absorption laws.

The $\max$ operation in turn induces a partial order on $\RR$ by defining
\[
  u \le v := \max(u, v) = v.
\]
Reflexivity, antisymmetry, and transitivity then follow from idempotence,
commutativity, and associativity of $\max$, respectively. With the definition of
$\le$ in place, \Cref{lemma:max-translation-invariant-real} has the expected
consequence that addition preserves and reflects order.

\begin{lemma}
  \label{lemma:addition-monotone-reflective}
  For each real $a : \RR$, the maps
  \[
    u \mapsto a + u
    \qquad\text{and}\qquad
    u \mapsto u + a
  \]
  are monotone and order-reflecting with respect to $\le$. Explicitly,
  \[
    u \le v \iff a + u \le a + v
  \]
  and hence also
  \[
    u \le v \iff u + a \le v + a
  \]
  for all $a, u, v : \RR$.
\end{lemma}
\begin{proof}
  By commutativity of addition, it suffices to prove the claim for addition on
  the left.

  ($\Longrightarrow$) Suppose $u \le v$. Then by definition, $\max(u, v) =
  v$. It follows by \Cref{lemma:max-translation-invariant-real} that
  \[
    \max(a + u, a + v) = a + \max(u, v) = a + v,
  \]
  and hence $a + u \le a + v$, again by the definition of $\le$.

  ($\Longleftarrow$) Suppose $a + u \le a + v$. Then by definition
  \[
    \max(a + u, a + v) = a + v.
  \]
  Applying \Cref{lemma:max-translation-invariant-real} with $-a$ in place of
  $a$, we have
  \begin{align*}
    \max(u, v)
      & = \max((-a) + (a + u), (-a) + (a + v)) \\
      & = (-a) + \max(a + u, a + v) \\
      & = (-a) + (a + v) \\
      & = v,
  \end{align*}
  and so $u \le v$.
\end{proof}

We now turn to the second case study, concerning strict order. The non-strict
order relation $\le$ was induced from the $\max$ operation, lifted from the
corresponding operation on rationals. In contrast, strict inequality cannot be
obtained directly from the lift of some algebraic operation. As a result, we
will not be able to immediately derive properties of the strict order relation
as a consequence of the uniqueness of continuous extensions; proving properties
of the strict order relation will require a bit more attention. Therefore, our
goal in this case study is not only to define $<$, but to develop enough theory
around it to use it in later algebraic arguments. The HoTT book proves several
lemmas about order and gives a characterization of closeness by distance, but it
does not for instance show that addition preserves or reflects strict inequality
\cite{univalentfoundationsprogram2013}. For that result, we follow
\textcite{gilbert2017} in deriving an alternative characterization of strict
inequality: namely, that $u < v$ holds exactly when there merely exists a
positive rational $\varepsilon$ such that
\[
  u + \rational(\varepsilon) \le v.
\]
The results below are organized around this characterization. We begin with the
definition of $<$ and its immediate consequences, then prove perturbation lemmas
relating $<$ to closeness, and finally derive Gilbert's characterization to show
that addition preserves and reflects strict inequality.

The definition of strict inequality given in
\cite{univalentfoundationsprogram2013} makes use of our intuition that the reals
should satisfy the Archimedean property. We can utilize our expectation that the
rationals are dense in the reals\footnote{This is equivalent to the more typical
  characterization that for any $u : \RR$ there exists an integer $k$ with
  $u < k$. See \cite[Exercise 11.7]{univalentfoundationsprogram2013}.} to
express strict inequality as
\[
  u < v := \exists q, r : \QQ, (u \le \rational(q)) \times (q < r) \times (\rational(r) \le v)
\]
where $q < r$ is strict inequality on the rationals
\cite[\S11.3.3]{univalentfoundationsprogram2013}.

The irreflexivity and transitivity of $<$ follow immediately from showing that
the $\rational$ constructor preserves and reflects $\le$ and then in turn
$<$. We also obtain proofs that $\le$ and $<$ satisfy
\begin{align*}
  u < v & \to v \le w \to u < w, \\
  u \le v & \to v < w \to u < w
\end{align*}
for all $u, v, w : \RR$. As indicated, the Archimedean property follows
immediately from the definition.

\begin{theorem}[Archimedean principle; {\textcite[Theorem 11.3.41]{univalentfoundationsprogram2013}}]
  \label{theorem:archimedean}
  For every $u, v : \RR$ with $u < v$, there merely exists $q : \QQ$ such that
  $u < \rational(q) < v$.
\end{theorem}
\begin{proof}
  Suppose $u < v$. Then by definition there exist $r, s : \QQ$ such that
  \[
    u \le \rational(r), \qquad r < s, \qquad \rational(s) \le v.
  \]
  Choose $q := \frac{r + s}{2}$, the midpoint of $r$ and $s$. Then $r < q < s$.

  First, we have
  \[
    u \le \rational(r), \qquad r < q, \qquad \rational(q) \le \rational(q),
  \]
  where the last inequality is by reflexivity of $\le$. Thus we obtain
  $u < \rational(q)$ by definition.
  
  Similarly, we have
  \[
    \rational(q) \le \rational(q), \qquad q < s, \qquad \rational(s) \le v,
  \]
  so again by definition of $<$ we obtain $\rational(q) < v$.
  
  Hence $u < \rational(q) < v$.
\end{proof}

This gives us the immediate consequences of the definition of strict
inequality. To make further use of $<$, however, we need to understand how it
interacts with the closeness relation. The next two lemmas from the HoTT book
describe how small perturbations interact with non-strict and strict rational
bounds. They are followed by an important characterization of closeness in terms
of distance, confirming that the inductively defined closeness relation
$\closesym$ is equivalent to the more familiar notion of closeness induced by
the distance metric $d(u, v) := \lvert u - v \rvert$ on the reals.

\begin{lemma}[{\textcite[Lemma 11.3.42]{univalentfoundationsprogram2013}}]
  \label{lemma:less-equal-rational-close-less-equal-rational-plus}
  Let $u, v : \RR$, $q : \QQ$ and suppose $u \le \rational(q)$. If
  $u \close{\varepsilon} v$ for some $\varepsilon : \Qp$ then
  $v \le \rational(q + \varepsilon)$.
\end{lemma}

\begin{lemma}[{\textcite[Lemma 11.3.43(i)]{univalentfoundationsprogram2013}}]
  \label{lemma:less-than-rational-close-less-than-rational-plus}
  Let $u, v : \RR$, $q : \QQ$ and assume $u < \rational(q)$. If
  $u \close{\varepsilon} v$ for some $\varepsilon : \Qp$ then
  $v < \rational(q + \varepsilon)$.
\end{lemma}

\begin{theorem}[{\textcite[Theorem 11.3.44]{univalentfoundationsprogram2013}}]
  \label{theorem:close-iff-distance-less-than}
  For all $u, v : \RR$ and all $\varepsilon : \Qp$, we have
  \[
    \left(u \close{\varepsilon} v\right) \iff \left(\left\lvert u - v \right\rvert < \rational(\varepsilon)\right).
  \]
\end{theorem}

Up to this point, we have followed the development for strict order given in
\cite{univalentfoundationsprogram2013}. We now turn to Gilbert's strategy for
proving that addition preserves and reflects strict inequality. The key idea is
to derive a second characterization of $<$ involving addition by a positive
rational
\[
  u < v \iff \exists \varepsilon : \Qp, u + \rational(\varepsilon) \le v.
\]
This reformulation is useful because it lets us build on the algebraic structure
already developed for non-strict inequality. Once $<$ is expressed in terms of
addition and $\le$, we can use the fact that addition already preserves and
reflects $\le$ to prove the corresponding result for strict inequality.

The route to this characterization proceeds in four steps. We first prove a
helper lemma showing that $\varepsilon$-closeness yields a non-strict order
bound of the form
\[
  v \le u + \rational(\varepsilon).
\]
From this we derive two perturbation lemmas for strict inequality, showing how
an inequality $u < v$ is affected when either endpoint is replaced by a
sufficiently close real. We then establish the weak linearity of $<$, which is
needed to prove that every real lies strictly below any positive rational
perturbation of itself. With these results in place, Gilbert's alternative
characterization of strict inequality follows.

We begin with the promised non-strict order bound arising from closeness.

\begin{lemma}[{\textcite[Lemma 4.2]{gilbert2017}}]
  \label{lemma:close-perturb-less-equal}
  For all $u, v : \RR$ and $\varepsilon : \Qp$, if $u \close{\varepsilon} v$
  then $v \le u + \rational(\varepsilon)$.
\end{lemma}
\begin{proof}
  Suppose $u \close{\varepsilon} v$. By
  \Cref{theorem:close-iff-distance-less-than}, this implies
  \[
    \lvert u - v \rvert < \rational(\varepsilon),
  \]
  which we can weaken to
  \[
    \lvert u - v \rvert \le \rational(\varepsilon).
  \]
  Since $- (u - v) \le \lvert u - v \rvert$, we have
  \[
    - u + v = - (u - v) \le \lvert u - v \rvert \le \rational(\varepsilon).
  \]
  By the monotonicity of addition (\Cref{lemma:addition-monotone-reflective}),
  adding $u$ to both sides yields
  \[
    v \le u + \rational(\varepsilon)
  \]
  which is the desired inequality.
\end{proof}

The next step is to promote this non-strict control to corresponding
perturbation results for strict inequality.

\begin{lemma}[{\textcite[Lemma 4.3]{gilbert2017}}]
  \label{lemma:less-than-close-less-than-add}
  For all $u, v, w : \RR$ and $\varepsilon : \Qp$, if $u < v$ and
  $u \close{\varepsilon} w$ then $w < v + \rational(\varepsilon)$.
\end{lemma}
\begin{proof}
  Suppose $u < v$ and that $u \close{\varepsilon} w$ for some
  $\varepsilon : \Qp$. By the Archimedean property (\Cref{theorem:archimedean}),
  there merely exists some $q : \QQ$ such that
  \[
    u < \rational(q) < v.
  \]
  Applying \Cref{lemma:less-than-rational-close-less-than-rational-plus} to
  $u < \rational(q)$ and $u \close{\varepsilon} w$, we obtain
  \[
    w < \rational(q + \varepsilon).
  \]
  Weakening $\rational(q) < v$ to $\rational(q) \le v$ and adding $\varepsilon$
  to both sides (\Cref{lemma:addition-monotone-reflective}) yields
  \[
    \rational(q + \varepsilon) \le v + \rational(\varepsilon).
  \]
  By the transitivity of $\le$ and $<$, it follows that
  \[
    w < v + \rational(\varepsilon)
  \]
  as needed.
\end{proof}

Although Gilbert does not formulate the next corollary separately, recording it
explicitly makes the subsequent proof of weak linearity a bit cleaner.

\begin{corollary}
  \label{corollary:less-than-close-subtract-less-than}
  For all $u, v, w : \RR$ and $\varepsilon : \Qp$, if $u < v$ and
  $v \close{\varepsilon} w$ then $u - \rational(\varepsilon) < w$.
\end{corollary}
\begin{proof}
  Assume $u < v$ and that $v \close{\varepsilon} w$ for some
  $\varepsilon : \Qp$. Negation is antitone (order reversing) with
  respect to $<$, so we have $- v < - u$. Negation is also non-expanding, so
  from $v \close{\varepsilon} w$ it follows that $- v \close{\varepsilon} -
  w$. Applying the previous lemma yields
  \[
    - w < - u + \rational(\varepsilon).
  \]
  Negating both sides gives
  \[
    u - \rational(\varepsilon) < w.
  \]
\end{proof}

The next result is a constructive substitute for the classical trichotomy law
for linear orders. Classically, for any two reals $u$ and $v$, exactly one of
\[
  u < v, \qquad u = v, \qquad v < u
\]
holds. Constructively, however, this is too strong: strict comparison of
arbitrary real numbers is not decidable in general. Weak
linearity\footnote{There seems to be some variation in terminology here. The
  HoTT book refers to the property
  $\forall x, y, z,\; R(x, y) \to R(x, z) \lor R(z, y)$ as \emph{weak
    linearity}, while it reserves \emph{cotransitivity} for the variant
  $\forall x, y, z,\; R(x, y) \to R(x, z) \lor R(y, z)$ (the difference being in
  the order of $R(z, y)$ versus $R(y, z)$). \textcite{gilbert2017} uses the
  latter term for the former property, while \textcite{booij2020} treats both
  terms as synonyms for the first property. The standard library for Cubical
  Agda follows the HoTT book, so we keep that convention here.} captures the
part of this classical intuition that remains available constructively. It
asserts that whenever $u < v$, any third real $w$ must lie on one side or the
other of that strict interval, in the sense that
\[
  u < w \lor w < v.
\]
As explained in \cite[\S11.2.1]{univalentfoundationsprogram2013}, taking
$u := z - \rational(\varepsilon)$ and $v := z + \rational(\varepsilon)$ yields
\[
  z - \rational(\varepsilon) < w \lor w < z + \rational(\varepsilon),
\]
and so we can view weak linearity as a form of linearity up to a small
error. Besides being an important structural property of $<$ in its own right,
it is also the key to proving \Cref{lemma:less-than-perturb} immediately
afterwards.

\begin{lemma}[{\textcite[Lemma 4.4]{gilbert2017}}]
  \label{lemma:less-than-weakly-linear}
  Strict inequality is weakly linear: for all $u, v, w : \RR$,
  \[
    u < v \implies u < w \lor w < v.
  \]
\end{lemma}
\begin{proof}
  It suffices to prove
  \[
    q < r \implies \rational(q) < w \lor w < \rational(r)
  \]
  for $q, r : \QQ$. Indeed, if $u < v$ then by definition there merely exist
  some $q, r : \QQ$ such that
  \[
    u \le \rational(q), \quad q < r, \quad \rational(r) \le v.
  \]
  Thus, if $\rational(q) < w$ or $w < \rational(r)$, then $u < w$ or $w < v$
  follows by transitivity.

  Let $q, r : \QQ$ with $q < r$. We proceed by $\RR$-induction on $w$.

  The rational case follows immediately from the corresponding property for
  rational strict inequality.

  For the limit case, let $x : \Qp \to \RR$ be a Cauchy approximation and assume
  inductively that
  \[
    s < t \implies \rational(s) < x_{\varepsilon} \lor x_{\varepsilon} < \rational(t).
  \]
  for all $s, t : \QQ$ with $s < t$ and all $\varepsilon : \Qp$.

  Define
  \begin{align*}
    s & := \left(1 - \frac{1}{3}\right) q + \frac{1}{3} r, \\
    t & := \left(1 - \frac{2}{3}\right) q + \frac{2}{3} r,
  \end{align*}
  so that $q < s < t < r$. Next define
  \begin{align*}
    \delta_{1} & := s - q, \\
    \delta_{2} & := r - t, \\
    \delta & := \frac{\min(\delta_{1}, \delta_{2})}{2}.
  \end{align*}
  Then $0 < \delta < \delta_{1}, \delta_{2}$. Applying the inductive
  hypothesis to $s < t$ and $x_{\delta}$ yields
  \[
    \rational(s) < x_{\delta} \lor x_{\delta} < \rational(t),
  \]
  and hence there are two cases.

  \begin{enumerate}
  \item Suppose $\rational(s) < x_{\delta}$. Since
    $x_{\delta} \close{(\delta_{1} - \delta) + \delta} \limit(x)$,
    \Cref{corollary:less-than-close-subtract-less-than} gives
    \[
      \rational(q) = \rational(s - \delta_{1}) < \limit(x).
    \]
  \item Suppose $x_{\delta} < \rational(t)$. Since
    $x_{\delta} \close{(\delta_{2} - \delta) + \delta} \limit(x)$, it follows by
    \Cref{lemma:less-than-close-less-than-add} that
    \[
      \limit(x) < \rational(t + \delta_{2}) = \rational(r).
    \]
  \end{enumerate}

  In either case, we obtain
  \[
    \rational(q) < \limit(x) \lor \limit(x) < \rational(r),
  \]
  as required.
\end{proof}

Weak linearity now allows us to prove that every real is strictly smaller than
any positive rational perturbation of itself, which is the final result needed
before the characterization theorem.

\begin{lemma}[{\textcite[Lemma 4.5]{gilbert2017}}]
  \label{lemma:less-than-perturb}
  For every $u : \RR$ and every $\varepsilon : \Qp$, we have
  $u < u + \rational(\varepsilon)$.
\end{lemma}
\begin{proof}
  We proceed by $\RR$-induction on $u$.

  For the rational case, note that rational addition respects strict inequality,
  so
  \[
    q = q + 0 < q + \varepsilon
  \]
  and hence the desired result
  \[
    \rational(q) < \rational(q) + \rational(\varepsilon)
  \]
  follows by the strict monotonicity of the $\rational$ constructor.

  For the limit case, let $x : \Qp \to \RR$ be a Cauchy approximation and
  suppose inductively
  \[
    x_{\eta} < x_{\eta} + \rational(\zeta)
  \]
  for all $\eta, \zeta : \Qp$.

  Let $\varepsilon : \Qp$. Define $\delta := \frac{\varepsilon}{5}$. By the
  inductive hypothesis, we have
  \[
    x_{\delta} < x_{\delta} + \rational(\delta).
  \]
  Since $x_{\delta} \close{\delta + \delta} \limit(x)$,
  \Cref{lemma:less-than-close-less-than-add} yields
  \begin{equation}
    \label{eq:or1}
    \limit(x) < (x_{\delta} + \rational(\delta)) + \rational(\delta + \delta) = x_{\delta} + \rational(3 \delta).
  \end{equation}
  Then, applying the weak linearity of $<$
  (\Cref{lemma:less-than-weakly-linear}) to \labelcref{eq:or1}, we obtain either
  \[
    \limit(x) < \limit(x) + \rational(\varepsilon),
  \]
  the desired result, or
  \begin{equation}
    \label{eq:foo}
    \limit(x) + \rational(\varepsilon) < x_{\delta} + \rational(3 \delta).
  \end{equation}

  We show that \labelcref{eq:foo} leads to a contradiction. Since
  $x_{\delta} \close{\delta + \delta} \limit(x)$,
  \Cref{lemma:close-perturb-less-equal} yields
  \[
    x_{\delta} \le \limit(x) + \rational(\delta + \delta).
  \]
  By \Cref{lemma:addition-monotone-reflective}, adding $3 \delta$ to both sides,
  we obtain
  \begin{align*}
    x_{\delta} + \rational(3 \delta)
    & \le \limit(x) + \rational(\delta + \delta) + \rational(3 \delta) \\
    & = \limit(x) + \rational(5 \delta) \\
    & = \limit(x) + \rational(\varepsilon).
  \end{align*}
  Since we assumed
  $\limit(x) + \rational(\varepsilon) < x_{\delta} + \rational(3 \delta)$, this
  implies
  \[
    \limit(x) + \rational(\varepsilon) < \limit(x) + \rational(\varepsilon),
  \]
  contradicting the irreflexivity of $<$.
\end{proof}

We can now assemble the preceding lemmas into the desired alternative
characterization of strict inequality.

\begin{theorem}[{\textcite[Lemma 4.1]{gilbert2017}}]
  \label{theorem:less-than-iff-perturb-less-equal}
  Let $u, v : \RR$. Then $u < v$ if and only if there merely exists some
  $\varepsilon : \Qp$ such that $u + \rational(\varepsilon) \le v$.
\end{theorem}
\begin{proof}
  ($\Longrightarrow$) Suppose $u < v$. By definition, there merely exist some
  $q, r : \QQ$ such that
  \[
    u \le \rational(q), \quad q < r, \quad \rational(r) \le v.
  \]
  Let $\varepsilon := r - q$. We have $\varepsilon > 0$ since $q < r$. By the
  monotonicity of addition with respect to $\le$
  (\Cref{lemma:addition-monotone-reflective}), we have
  \[
    u + \rational(\varepsilon) \le \rational(q) + \rational(\varepsilon) =
    \rational(r) \le v.
  \]

  ($\Longleftarrow$) Suppose there merely exists some $\varepsilon : \Qp$ such
  that $u + \rational(\varepsilon) \le v$. By \Cref{lemma:less-than-perturb}, we
  obtain
  \[
    u < u + \rational(\varepsilon) \le v,
  \]
  so the result follows by the transitivity of $<$ and $\le$.
\end{proof}

This alternative characterization is the key result supplied by Gilbert's
strategy. The original definition of $u < v$ is expressed using the existence of
rationals separating $u$ and $v$, but
\Cref{theorem:less-than-iff-perturb-less-equal} shows that this is equivalent to
a formulation involving only addition and the non-strict order $\le$. Since we
previously established that addition preserves and reflects $\le$ in
\Cref{lemma:addition-monotone-reflective}, the desired compatibility of addition
with strict inequality now follows formally.

\begin{lemma}
  \label{lemma:addition-strict-monotone-reflective}
  For each real $a : \RR$, the maps
  \[
    u \mapsto a + u
    \qquad\text{and}\qquad
    u \mapsto u + a
  \]
  are monotone and order-reflecting with respect to $<$. Explicitly,
  \[
    u < v \iff a + u < a + v
  \]
  and hence also
  \[
    u < v \iff u + a < v + a
  \]
  for all $a, u, v : \RR$.
\end{lemma}
\begin{proof}
  By commutativity of addition, it suffices to prove the claim for addition on
  the left.

  ($\Longrightarrow$) Suppose $u < v$. Then by
  \Cref{theorem:less-than-iff-perturb-less-equal}, there merely exists some
  $\varepsilon : \Qp$ with $u + \rational(\varepsilon) \le v$. It follows by the
  monotonicity of addition with respect to $\le$ that
  \[
    (a + u) + \rational(\varepsilon) \le a + v,
  \]
  and hence $a + u < a + v$, again by
  \Cref{theorem:less-than-iff-perturb-less-equal}.

  ($\Longleftarrow$) Conversely, suppose $a + u < a + v$. Applying the
  implication just proved by adding $-a$ to both sides, we get
  \[
    (-a) + (a + u) < (-a) + (a + v).
  \]
  Simplifying both sides gives $u < v$.
\end{proof}

Before turning to the construction of multiplication and reciprocal, we record
one further consequence of the theory for strict order developed above. For
reciprocal, it is not enough to know merely that a real is not equal to zero;
rather, we need positive information that it lies on one side of zero or the
other. This is expressed by an apartness relation of $\RR$, defined as the
symmetric closure of $<$:
\[
  u \apart v := (u < v) + (v < u).
\]
Here we use the ordinary coproduct $+$, rather than propositionally truncated
logical disjunction $\lor$. The coproduct itself is already a proposition, since
the two injections are mutually exclusive and both propositionally-valued. Thus
a witness of $x \apart 0$ can be analyzed by cases, giving exactly the
computational content needed to define reciprocal separately for positive and
negative reals. Because $<$ is a strict order, its symmetric closure is
irreflexive, symmetric, and cotransitive, so $\apart$ is an apartness relation
in the standard constructive sense \cite{minesrichmanruitenburg2012}.

For the third case study, we focus on the construction of multiplication and
reciprocal on $\RR$. Unlike addition, negation, $\min$, and $\max$, these
operations cannot be obtained by a direct application of the extension
principles of \Cref{section:lifting}. The difficulty is that neither operation
is globally Lipschitz: multiplication admits no uniform Lipschitz constant in
one variable without a bound on the other, while reciprocal fails to be
Lipschitz near zero. The HoTT book addresses this problem indirectly by first
constructing the squaring operation $u \mapsto u^{2}$ and then defining
multiplication using the identity
\[
  u \cdot v = \frac{(u + v)^{2} - u^{2} - v^{2}}{2}
\]
\cite{univalentfoundationsprogram2013}. We instead follow a more direct strategy
due to Gilbert \cite{gilbert2017}. The basic idea is to define each operation
first on restricted domains where the missing Lipschitz bounds become available
again. For multiplication, this means fixing a bound on one factor, while for
reciprocal it means restricting to reals bounded below by a positive
rational. The main work is then to show that these locally defined maps agree
whenever two such bounds overlap. In other words, the output of the local map
must be independent of the particular chosen bound. Combined with the mere
existence of suitable bounds, we can pass to the desired global maps. To
accomplish this formally, Gilbert appeals to a principle of definition by
surjection, which may be viewed as a consequence of the universal property of
set quotients. We follow the same mathematical strategy, but adapt its
implementation to the infrastructure already available in the Cubical Agda
standard library. In particular, we make use of Kraus's characterization of maps
out of propositional truncations into sets \cite{kraus2015}. We describe the
construction of multiplication in some detail, and then summarize the
construction of reciprocal, since it follows the same general pattern.

We begin with multiplication on the left by a fixed rational. In general, a
Lipschitz constant bounds the factor by which a map can expand distances, so for
multiplication by a fixed scalar $q : \QQ$ we expect multiplication by $q$ to be
Lipschitz with constant $\lvert q \rvert$. We take the absolute value because
distances are nonnegative, and to fit the HoTT book's convention that
Lipschitz constants be positive rationals (see \Cref{definition:lipschitz}), we
replace it with the positive bound $\max(\lvert q \rvert, 1)$.

\begin{lemma}
  \label{lemma:left-multiply-rational-lipschitz}
  Let $q : \QQ$. Then the map
  \[
    r \mapsto \rational(q \cdot r) : \QQ \to \RR
  \]
  is Lipschitz with constant $\max(\lvert q \rvert, 1)$.
\end{lemma}
\begin{proof}
  Suppose $r, s : \QQ$ satisfy $\lvert r - s \vert < \varepsilon$ for some
  $\varepsilon : \Qp$. Then
  \[
    \lvert q \cdot r - q \cdot s \rvert
    = \lvert q \rvert \cdot \lvert r - s \rvert
    \le \max(\lvert q \rvert, 1)\cdot \lvert r - s \rvert
    < \max(\lvert q \rvert, 1)\cdot \varepsilon.
  \]
  Hence
  \[
    \rational(q \cdot r)
    \close{\max(\lvert q \rvert, 1) \cdot \varepsilon}
    \rational(q \cdot s),
  \]
  so the displayed map is Lipschitz with constant $\max(\lvert q \rvert, 1)$.
\end{proof}

Applying the Lipschitz extension principle (\Cref{lemma:lift-lipschitz}) to
\Cref{lemma:left-multiply-rational-lipschitz}, we obtain a function
\[
  m : \QQ \to \RR \to \RR
\]
where for each $q : \QQ$ the map $m_{q} : \RR \to \RR$ extends left
multiplication by $q$. In particular,
\[
  m_{q}(\rational(r)) = \rational(q \cdot r)
\]
for all $q, r : \QQ$. Since each $m_{q}$ is obtained by Lipschitz extension, it
is again Lipschitz with constant $\max(\lvert q \rvert, 1)$. This completes the
first extension step, extending multiplication by a fixed rational from $\QQ$ to
$\RR$. To extend in the remaining argument, we now fix a real $v : \RR$ and ask
whether the map
\[
  q \mapsto m_{q}(v)
\]
is Lipschitz in the rational parameter $q$. This is not true uniformly in $v$,
but it does hold once $\lvert v \rvert$ is bounded by a positive rational.

\begin{lemma}
  \label{lemma:m-lipschitz1}
  Let $L : \QQ$ be positive, and let $v : \RR$ satisfy
  $\lvert v \rvert \le \rational(L)$. Then the map
  \[
    q \mapsto m_{q}(v) : \QQ \to \RR
  \]
  is $L$-Lipschitz.
\end{lemma}
\begin{proof}
  We use two properties of the family $m_{q}$, both obtained by proving the
  corresponding rational identities and then lifting them to the reals by
  uniqueness of continuous extensions. First, for all $q, r : \QQ$ and
  $u : \RR$,
  \[
    \lvert m_{q}(u) - m_{r}(u) \rvert = m_{\lvert q - r \rvert}(\lvert u \rvert).
  \]
  Second, for every nonnegative rational $s$ and all $u, w : \RR$,
  \[
    \max(m_{s}(u), m_{s}(w)) = m_{s}(\max(u, w)).
  \]
  The second identity implies that $m_{s}$ is monotone whenever $s \ge 0$.

  Now let $q, r : \QQ$, let $\varepsilon : \Qp$, and assume
  \[
    \lvert q-r\rvert < \varepsilon.
  \]
  Then
  \[
    \lvert m_{q}(v) - m_{r}(v) \rvert =
    m_{\lvert q - r \rvert}(\lvert v \rvert) \le
    m_{\lvert q - r \rvert}(\rational(L)) =
    \rational(\lvert q - r \rvert L)
  \]
  by the first property, monotonicity of $m_{\lvert q - r \rvert}$, the bound
  $\lvert v \rvert \le \rational(L)$, and the fact that $m$ computes on rational
  inputs.

  By \Cref{theorem:close-iff-distance-less-than}, it follows that the map
  $q \mapsto m_{q}(v)$ is $L$-Lipschitz.
\end{proof}

Applying \Cref{lemma:lift-lipschitz} once more, we obtain that whenever
$L : \Qp$ and $v : \RR$ satisfy $\lvert v \rvert \le \rational(L)$, the map
\[
  q \mapsto m_{q}(v)
\]
extends from $\QQ$ to a map
\[
  b_{L, v} : \RR \to \RR
\]
As a consequence of the Lipschitz extension principle, this bounded
multiplication map satisfies
\[
  b_{L, v}(\rational(q)) = m_{q}(v)
\]
for all $q : \QQ$ and is Lipschitz with constant $L$. Thus, once a
rational bound on $\lvert v \rvert$ has been chosen, right multiplication by $v$
is defined on all real inputs.

This completes the local part of the construction of multiplication. To pass
from the bound-indexed local maps $b_{L, v}$ to genuine multiplication defined
on all of $\RR$, two requirements must be met. First, suitable rational bounds
must merely exist for every real. Second, whenever two such bounds are both
valid, the corresponding bounded multiplication maps must agree. The next lemma
supplies the first of these requirements. 

\begin{lemma}[{\textcite[Lemma 4.13]{gilbert2017}}]
  \label{lemma:bounded-by-rational}
  For every $u : \RR$, there merely exists a $q : \Qp$ such that $\lvert u \rvert < \rational(q)$.
\end{lemma}
\begin{proof}
  Applying \Cref{lemma:less-than-perturb} to $\lvert u \rvert$ with
  $\varepsilon := 1$, we obtain
  \[
    \lvert u \rvert < \lvert u \rvert + 1.
  \]
  By the Archimedean property (\Cref{theorem:archimedean}), there merely
  exists some $q : \QQ$ such that
  \[
    \lvert u \rvert < \rational(q) < \lvert u \rvert + 1.
  \]
  Since $\lvert u \rvert \ge 0$, it follows that $0 < \rational(q)$. Therefore
  there merely exists $q : \Qp$ such that
  \[
    \lvert u \rvert < \rational(q).
  \]
\end{proof}

Our remaining goal is to fulfill the second requirement by eliminating the
dependence on the chosen bound. We briefly compare two ways of making this
precise: Gilbert's principle of definition by surjection and the
characterization of maps of the form $\bracket{A} \to B$ where $B$ is a set due
to Kraus, which is what we use in the formalization.

Following Kraus, a map $f : A \to B$ is \textbf{weakly constant} if it comes
equipped with an element of type
\[
  \mathsf{IsWeaklyConstant}(f) := \prd{x, y : A} f(x) = f(y).
\]
In general, weak constancy is not enough to define a map out of a proposition
truncation, since the paths
\[
  f(x) = f(y)
\]
may themselves need to satisfy higher coherence conditions. When the codomain
$B$ is a set, however, all such higher conditions are automatic, so weak
constancy is already sufficient.

\begin{theorem}[{\textcite[Proposition 2.2]{kraus2015}}]
  \label{theorem:kraus}
  If $A$ is a type and $B$ is a set, there is an equivalence
  \[
    (\bracket{A} \to B) \quad \simeq \quad \sm{f : A \to B} \mathsf{IsWeaklyConstant}(f).
  \]
\end{theorem}

The consequence of \Cref{theorem:kraus} that we use here is that every weakly
constant map $f : A \to B$ into a set $B$ extends uniquely to a map
$\overline{f} : \bracket{A} \to B$. This is summarized by the commutative
diagram
\[
  \begin{tikzcd}[column sep=huge, row sep=huge]
    \bracket{A} \arrow[r, dashed, "\overline{f}"] & B \\
    A \arrow[u, "\eta"] \arrow[ur, "f"']
  \end{tikzcd}
\]
where $\eta : A \to \bracket{A}$ is the canonical map into the propositional
truncation. For this approach, the relevant data consists of a map
$f : A \to B$, the required invariance is weak constancy, and the output is the
extension $\overline{f} : \bracket{A} \to B$.

Gilbert packages the same basic idea in a different way, using a principle of
\textbf{definition by surjection} \cite[Definition 4.9]{gilbert2017}. Suppose
$f : A \to C$ and $g : A \to B$ are maps such that $g$ is surjective and $f$
respects the equivalence relation on $A$ induced by $g$, defined by
\[
  a_{1} \sim_{g} a_{2} := g(a_{1}) = g(a_{2}).
\]
Under these hypotheses, the universal property of the set quotient by $\sim_{g}$
yields a map
\[
  \overline{f} : A / {\sim_{g}} \to C
\]
extending $f$. The same universal property also gives a map
\[
  \overline{g} : A / {\sim_{g}} \to B
\]
extending $g$, and the surjectivity of $g$ implies that $\overline{g}$ is an
equivalence. We may therefore define
\[
  h := \overline{f} \circ \overline{g}^{-1} : B \to C,
\]
which satisfies $h \circ g = f$. This setup is represented by the commutative
diagram
\[
  \begin{tikzcd}[column sep=huge, row sep=huge]
    B & A \arrow[l, "g"'] \arrow[d, "\pi"] \arrow[r, "f"] & C \\
    & A / {\sim_{g}} \arrow[ul, dashed, "\overline{g}"] \arrow[ur,
    dashed, "\overline{f}"'] &
  \end{tikzcd}
\]
where $\pi : A \to A / {\sim_{g}}$ is the canonical projection onto the
quotient. Here the relevant data are the maps $f : A \to C$ and $g : A \to B$,
the required invariance is that $f$ respects the equivalence relation induced by
$g$, and the output is the induced map $h : B \to C$.

We now show how both abstract principles can be instantiated in the construction
of multiplication, thereby eliminating the dependence on the chosen bound. To
separate the real inputs from the chosen bound on the second argument, for each
$u, v : \RR$, we define a map
\begin{align*}
  & h_{u, v} : \left(\sm{L : \Qp} \lvert v \rvert \le \rational(L)\right) \to \RR \\
  & h_{u, v}(L, \varphi, \psi) := b_{L, v}(u).
\end{align*}

Since $\RR$ is a set by \Cref{corollary:reals-set}, Kraus's theorem applies
directly to this situation. Once $h_{u, v}$ is shown to be weakly constant, it
extends uniquely to a map
\[
  \overline{h}_{u, v} : \bracket{\sm{L : \Qp} \lvert v \rvert \le \rational(L)} \to \RR
\]
as summarized by the commutative diagram:
\[
  \begin{tikzcd}[column sep=huge, row sep=huge]
    \bracket{\sm{L : \Qp} \lvert v \rvert \le \rational(L)}
    \arrow[r, dashed, "\overline{h}_{u,v}"] & \RR \\
    \sm{L : \Qp} \lvert v \rvert \le \rational(L)
    \arrow[u, "\eta"]
    \arrow[ur, "h_{u,v}"']
  \end{tikzcd}
\]
By \Cref{lemma:bounded-by-rational}, the propositional truncation
\[
  \bracket{\sm{L : \Qp} \lvert v \rvert \le \rational(L)}
\]
is always inhabited. Evaluating $\overline{h}_{u, v}$ at the truncated witness
therefore yields multiplication.

We can achieve the same result using Gilbert's approach. Fix $u : \RR$ and
define
\begin{align*}
  & A_u := \sm{v : \RR} \sm{L : \Qp} \lvert v \rvert \le \rational(L), \\
  & f_u  : A_u \to \RR, \\
  & f_u(v, L, \varphi, \psi) := b_{L, v}(u), \\
  & g_u  : A_u \to \RR, \\
  & g_u(v, L, \varphi, \psi) := v.
\end{align*}
By \Cref{lemma:bounded-by-rational}, the map $g_{u}$ is surjective. Since
$g_{u}$ forgets the chosen bound, proving that $f_{u}$ respects $\sim_{g_{u}}$
amounts to showing that $f_{u}$ is invariant with respect to the choice of
bound. Once this is known, we obtain a map
\[
  \mu_{u} : \RR \to \RR
\]
as summarized by the commutative diagram
\[
  \begin{tikzcd}[column sep=huge, row sep=huge]
    \RR &
    A_u \arrow[l, "g_u"'] \arrow[d, "\pi"] \arrow[r, "f_u"] & \RR \\
    & A_u / {\sim_{g_u}}
    \arrow[ul, dashed, "\overline{g}_u"]
    \arrow[ur, dashed, "\overline{f}_u"'] &
  \end{tikzcd}
\]
Multiplication is then given by
\[
  u \cdot v := \mu_{u}(v).
\]

Both principles therefore depend exactly on the two requirements stated above:
first, that every real is merely bounded by a positive rational, and second,
that bounded multiplication is invariant with respect to the chosen bound. Since
Kraus's formulation is more direct and is already available in the Cubical Agda
library, it is the one we adopt here. It therefore remains to establish the
required invariance.

\begin{lemma}
  \label{lemma:h-weakly-constant}
  For every $u, v : \RR$, the map
  \[
    h_{u, v} : \left( \sm{L : \Qp} \lvert v \rvert \le \rational(L) \right) \to \RR
  \]
  is weakly constant. Explicitly, if
  $(L, \varphi, \psi), (M, \omega, \chi) : \sm{L : \Qp} \lvert v \rvert \le
  \rational(L)$ then
  \[
    h_{u, v}(L, \varphi, \psi) = b_{L, v}(u) =
    b_{M, v}(u) = h_{u, v}(M, \omega, \chi).
  \]
\end{lemma}
\begin{proof}
  Both $b_{L, v}$ and $b_{M, v}$ are continuous maps $\RR \to \RR$, since each
  is obtained by Lipschitz extension. Moreover, they agree on rational inputs:
  for every $q : \QQ$,
  \[
    b_{L,v}(\rational(q)) = m_q(v) = b_{M,v}(\rational(q)).
  \]
  By \Cref{lemma:continuous-extension-unique}, it follows that
  \[
    b_{L,v} = b_{M,v}
  \]
  as functions $\RR \to \RR$. Evaluating both sides at $u$ yields
  \[
    b_{L,v}(u)=b_{M,v}(u),
  \]
  so $h_{u,v}$ is weakly constant.
\end{proof}

Combining \Cref{lemma:bounded-by-rational} with the previous lemma, and using
that $\RR$ is a set, applying Kraus's theorem to the map $h_{u, v}$ yields a
global multiplication operation
\[
  \cdot : \RR \to \RR \to \RR.
\]
For any $L : \Qp$ and $u, v : \RR$ such that $\lvert v \rvert \le \rational(L)$,
multiplication satisfies
\[
  u \cdot v = b_{L, v}(u).
\]
In particular, multiplication computes on rational inputs as expected: for all
$q, r : \QQ$,
\[
  \rational(q) \cdot \rational(r) = \rational(q \cdot r).
\]
The agreement with the bounded multiplication map also transfers the Lipschitz
estimate established above for the map $u \mapsto b_{L, v}(u)$.

\begin{lemma}[{\textcite[Lemma 4.18]{gilbert2017}}]
  \label{lemma:multiply-lipschitz1}
  Let $L : \Qp$ and let $v : \RR$ satisfy
  \[
    \lvert v \rvert \le \rational(L).
  \]
  Then the map
  \[
    u \mapsto u \cdot v
  \]
  is Lipschitz with constant $L$.
\end{lemma}
\begin{proof}
  By construction, the map $u \mapsto u \cdot v$ agrees with $b_{L,v}$. Since
  $b_{L,v}$ was obtained by Lipschitz extension of the $L$-Lipschitz map
  $q \mapsto m_q(v)$, it is itself $L$-Lipschitz.  Therefore
  $u \mapsto u \cdot v$ is $L$-Lipschitz as well.
\end{proof}

We next show that multiplication is continuous in each argument
separately. Continuity in the first argument is immediate from the previous
lemma.

\begin{lemma}[{\textcite[Lemma 4.19]{gilbert2017}}]
  \label{lemma:multiply-continuous1}
  For all $v : \RR$, the map
  \[
    u \mapsto u \cdot v : \RR \to \RR
  \]
  is continuous.
\end{lemma}
\begin{proof}
  By \Cref{lemma:bounded-by-rational}, there merely exists a positive rational
  $L$ such that
  \[
    \lvert v \rvert \le \rational(L).
  \]
  For any such $L$, \Cref{lemma:multiply-lipschitz1} shows that the map
  $u \mapsto u \cdot v$ is $L$-Lipschitz, and therefore continuous.
\end{proof}

Continuity in the second argument is less direct. The construction of
multiplication proceeded by fixing a bound on the second factor and extending in
the first, so a separate argument is needed for continuity in the second
variable.

\begin{lemma}
  \label{lemma:multiply-continuous2}
  For every $u : \RR$, the map
  \[
    v \mapsto u \cdot v : \RR \to \RR
  \]
  is continuous.
\end{lemma}
\begin{proof}
  Our argument rests on the identity
  \[
    \lvert a \cdot u - a \cdot v \rvert = \lvert a \rvert \cdot \lvert u - v \rvert,
  \]
  which holds for all $a, x, y : \RR$ and is obtained by lifting the
  corresponding rational formula using the uniqueness of continuous
  extensions. Note that, in verifying the continuity hypotheses needed for that
  lifting step, we only use the continuity of multiplication in the first
  argument.

  Fix $v : \RR$ and $\varepsilon : \Qp$. By \Cref{lemma:bounded-by-rational},
  there merely exists a positive rational $\eta$ such that
  $ \lvert u \rvert \le \rational(\eta)$. Choose
  $\delta := \frac{\varepsilon}{2 \eta}$. If $v \close{\delta} w$ then
  \Cref{theorem:close-iff-distance-less-than} gives
  \[
    \lvert v - w \rvert < \rational(\delta).
  \]
  Hence
  \[
    \lvert u \cdot v - u \cdot w \rvert
    = \lvert u \rvert \cdot \lvert v - w \rvert
    \le \rational(\eta)\cdot \rational(\delta)
    = \rational(\varepsilon / 2)
    < \rational(\varepsilon).
  \]
  Applying \Cref{theorem:close-iff-distance-less-than} again, we conclude
  $u \cdot v \close{\varepsilon} u \cdot w$. Therefore $v \mapsto u \cdot v$ is
  continuous.
\end{proof}

Because multiplication agrees with rational multiplication on rational inputs
and is continuous in each variable separately, commutativity, associativity, the
unit laws, and distributivity over addition lift from $\QQ$ to $\RR$. Hence
$\RR$ carries the structure of a commutative ring.

Reciprocal is constructed by the same local-to-global strategy, so we only
sketch the main ideas. The difference is that the relevant local domains are no
longer determined by upper bounds on magnitude, but by positive lower bounds
away from zero. Fix $\delta : \Qp$ and consider the rational map
\[
  q \mapsto \frac{1}{\max(q, \delta)}.
\]
Because the denominator is bounded below by $\delta$, this map is Lipschitz with
constant $\frac{1}{\delta^{2}}$, and therefore extends to a map
\[
  r_{\delta} : \RR \to \RR.
\]
Intuitively, $r_{\delta}$ is the reciprocal function clamped to the interval
$[\delta, \infty)$.

As with multiplication, the key step is to verify that these bounded maps agree
whenever two lower bounds are both valid. Specifically, if $\delta_{1}$ and
$\delta_{2}$ are two positive rational lower bounds that are both valid for the
same real $u$, then the corresponding local maps agree at $u$. Since every
positive real is merely bounded below by some positive rational, the same
truncation argument used for multiplication yields a globally defined reciprocal
on the positive reals.

Finally, the apartness relation introduced at the end of the previous case study
allows us to extend this positive reciprocal to all reals apart from zero. A
witness of $u \apart 0$ provides a case split into $0 < u$ or $u < 0$. In the
positive case we use the construction above; in the negative case we reduce to
the positive reciprocal of $-u$ and then negate. The resulting operation can be
shown to satisfy the property that every real is apart from zero if and only if
it is invertible.

We can now assemble the algebraic and order-theoretic structure developed in the
three case studies into the constructive definition of an Archimedean ordered
field used in the HoTT book.

\begin{definition}[{\textcite[Definition 11.2.7]{univalentfoundationsprogram2013}}]
  \label{definition:ordered-field}
  An \textbf{ordered field} consists of a set $F$ together with constants $0$,
  $1$, operations $-$, $+$, $\cdot$, $\min$, $\max$, and mere relations $\le$,
  $<$, $\apart$ such that:
  \begin{itemize}
    \item $(F, 0, 1, +, -, \cdot)$ is a commutative ring.
    \item $(F, \le, \min, \max)$ is a lattice.
    \item The relation $<$ is a strict order, meaning it is irreflexive, transitive, and weakly linear.
    \item The relation $\apart$ is an apartness relation, meaning it is irreflexive, symmetric, and cotransitive.
    \item Every element $x : F$ is invertible if and only if $x \apart 0$.
    \item For every $x, y, z : F$:
      \[
        \begin{array}{r c l @{\qquad} r c l}
          x \le y      & \iff     & \neg (y < x), &
                                                    x < y \le z  & \implies & x < z, \\
          x \apart y   & \iff     & (x < y) + (y < x), &
                                                         x \le y < z  & \implies & x < z, \\
          x \le y      & \iff     & x + z \le y + z, &
                                                       (x \le y) \times (0 \le z) & \implies & x z \le y z, \\
          x < y        & \iff     & x + z < y + z, &
                                                     0 < z        & \implies & (x < y \iff x z < y z), \\
          0 < x + y    & \implies & 0 < x \lor 0 < y, &
                                                        0            & <        & 1.
        \end{array}
      \]
  \end{itemize}
  Every ordered field comes equipped with a canonical embedding $r : \QQ \to
  F$. An ordered field is \textbf{Archimedean} if for every $x, y : F$ with
  $x < y$ there exists a $q : \QQ$ such that $x < r(q) < y$.
\end{definition}

\begin{theorem}
  \label{theorem:reals-archimedean-ordered-field}
  The HoTT book reals form an Archimedean ordered field.
\end{theorem}

The preceding case studies establish the main requirements of this theorem. The
operations of negation, addition, multiplication, together with their algebraic
laws, give the set $\RR$ the structure of a commutative ring, while $\min$ and
$\max$ induce the non-strict order relation $\le$ and satisfy the lattice
laws. The relation $<$ was shown to be a strict order satisfying the Archimedean
property, and both $\le$ and $<$ were shown to be preserved and reflected by
addition. Apartness was defined as the symmetric closure of $<$, and reciprocal
was constructed on exactly those reals apart from zero. The remaining
compatibility laws not discussed explicitly are verified in the Agda
formalization and are established using the same main proof patterns illustrated
in this chapter.

\chapter{Real numbers in Cubical Agda}
\label{chapter:reals-agda}

This chapter describes the Cubical Agda implementation of the mathematics
from~\Cref{chapter:reals-math}.
The implementation is the key technical contribution of this thesis: it
formalizes all definitions, theorems, and lemmas from
the previous chapter, as well as their proofs, as machine-checked code (over
13,000 lines).
The code is open source
and available at \url{https://github.com/utahplt/hott-reals}.

The chapter is organized as follows. \Cref{section:computation} gives a brief
example showing that arithmetic on the reals computes definitionally when
applied to rational inputs. \Cref{section:organization} surveys the structure of
the codebase, following the organization of \Cref{chapter:reals-math}.
\Cref{section:claude-code} describes our use of Claude Code during the
development. \Cref{section:lessons-learned} reflects on three insights that
emerged during the formalization. We do not walk through Agda proofs in detail;
the repository is the source of truth for the code itself.

\section{A Computational Example}
\label{section:computation}

Because the operations on the HoTT book reals are defined by recursion on the
higher inductive type, and the $\rational$ constructor computes definitionally,
arithmetic identities between rational real numbers hold by definition. For
example, the identity $\rational(2) + \rational(2) = \rational(4)$ is proved by
\texttt{refl}. The type checker verifies this directly.

\bigskip
\begingroup\setstretch{1.0}\makeatletter\renewcommand{\verbatim@font}{\juliamono}\makeatother
\begin{verbatim}
    2+2≡4 : (rational 2) + (rational 2) ≡ rational 4
    2+2≡4 = refl
\end{verbatim}
\endgroup
In Lean~4 (v4.29.0) with Mathlib (v4.29.0), the analogous statement
does not hold definitionally.

\bigskip
\begingroup\setstretch{1.0}\makeatletter\renewcommand{\verbatim@font}{\juliamono}\makeatother
\begin{verbatim}
    theorem two_add_two_equal_four_real : (2 : ℝ) + (2 : ℝ) = (4 : ℝ) := rfl
    -- Not a definitional equality: the left-hand side
    --   2 + 2
    -- is not definitionally equal to the right-hand side
    --   4
\end{verbatim}
\endgroup
This is not optimal. The construction of the reals in Lean uses a
quotient of Cauchy sequences in a way that does not support direct
normalization.

\section{Organization}
\label{section:organization}

The Cubical Agda implementation comprises 33 modules and approximately 13,500
lines of code. Its organization parallels that of \Cref{chapter:reals-math},
together with supporting infrastructure for the rationals that was not available
in the Cubical Agda standard library. \Cref{table:agda-organization} summarizes
the distribution of code by topic. In this section we describe our use of the
Cubical Agda standard library and the rational support infrastructure, and then
survey the code corresponding to each section of \Cref{chapter:reals-math}.

\begin{table}[!htbp]
  \centering
  \begingroup
  \setstretch{1.0}
  \begin{tabular}{>{\raggedright\arraybackslash}p{0.26\textwidth} >{\raggedright\arraybackslash}p{0.50\textwidth} r}
    \toprule
    Topic & Main modules & Lines \\
    \midrule
    Rational support infrastructure &
    \texttt{Data.Rationals.Order}, \texttt{Data.Rationals.Properties} &
    2,293 \\
    Definition, induction, and recursion &
    \texttt{Data.Real.Base}, \texttt{Data.Real.Definitions},
    \texttt{Data.Real.Induction} &
    679 \\
    Closeness &
    \texttt{Data.Real.Close.ReflexiveSymmetric},
    \texttt{Data.Real.Close.CloseAlternative},
    \texttt{Data.Real.Close.Other} &
    2,484 \\
    Continuity and extension &
    \texttt{Data.Real.Lipschitz.Base},
    \texttt{Data.Real.Lipschitz.Closed},
    \texttt{Data.Real.Nonexpanding},
    \texttt{Data.Real.Properties} &
    1,887 \\
    Algebra &
    \texttt{Data.Real.Algebra.*} &
    2,922 \\
    Order &
    \texttt{Data.Real.Order.*} &
    3,295 \\
    \midrule
    Total & & 13,560 \\ 
    \bottomrule
  \end{tabular}
  \caption{Distribution of the Agda codebase by topic.}
  \label{table:agda-organization}
  \endgroup
\end{table}

The formalization is built on the Cubical Agda standard library, which supplies
the foundations for doing formalized mathematics in Cubical Type
Theory~\cite{agdacubical}. We make extensive use of its propositional truncation
module, including eliminators at various arities and the \texttt{SetElim}
module, which provides the Cubical library's version of Kraus's characterization
of maps from propositional truncations into sets~\cite{kraus2015}. This is used
in the construction of both multiplication and reciprocal (see
\Cref{section:algebra-order}). The library also provides the theorem that any
type equipped with a reflexive mere relation implying identity is a
set~\cite[Theorem 7.2.2]{univalentfoundationsprogram2013}, which we use to show
that $\RR$ is a set. For order theory, we rely on the library's vocabulary of
posets, strict orders, and apartness relations, including the result that the
symmetric closure of a strict order is an apartness relation. On the algebraic
side, the library supplies definitions for groups, rings, and fields, which we
use to package the algebraic structure of $\QQ$ and $\RR$. The library's notion
of field is that of a denial field in the sense of
\textcite{minesrichmanruitenburg2012}, requiring inverses for elements
satisfying $\neg (x = 0)$ rather than elements apart from zero. Since the
ordered field definition of the HoTT book (\Cref{definition:ordered-field}) uses
an apartness relation, we were unable to use the library's field definition for
$\RR$. We did, however, use the standard definition for $\QQ$.

At the start of the project, the standard library did not include an instance
showing that $\QQ$ is a denial field, nor did it provide an explicit reciprocal
operator for the rationals. We contributed a rational denial field instance as a
pull request (\url{https://github.com/agda/cubical/pull/1260}) to the Cubical
Agda standard library, which was later accepted.

The supporting infrastructure for the rationals totals approximately 2,300 lines
across two modules. The Cubical Agda standard library provides the basic
definitions and algebraic operations for the rationals, but at the time of
writing there are many missing definitions and results that our formalization
requires. We therefore develop the missing results ourselves:
\begin{itemize}
  \item We formulate
missing interactions between the algebraic and order-theoretic structure, namely
the monotonicity and order-reflection of addition and multiplication, the fact
that negation is antitone, and the fact that reciprocal on positive rationals is
antitone, all with respect to both $\le$ and $<$.
    \item We build the lattice structure
on $\QQ$, showing that $\max$ and $\min$ form join- and meet-semilattices
respectively, and establish the midpoint identities for $\max$ and $\min$
mentioned in \Cref{lemma:max-midpoint-half-distance}.
    \item We define the absolute
value and distance functions on the rationals and develop their basic theory,
including the triangle inequality and reverse triangle inequality in both norm
and distance forms, the metric axioms for distance, and the norm axioms for
absolute value.
\end{itemize}
With these contributions in place, we prove that the rational versions of each
operation defined on the reals are Lipschitz or non-expanding. We also develop
lemmas about the closeness relation on the rationals, including the openness and
separatedness properties. Other supporting results include strict monotonicity
and bounds for affine combinations, clamping operations with interval
membership, and midpoint inequalities such as $q < \frac{q + r}{2} <
r$. Finally, we prove the rational identities that are subsequently lifted to
$\RR$ to establish the Archimedean ordered field axioms.


The code corresponding to \Cref{section:definition-induction-recursion} lives in
three modules. The \texttt{Base} module defines the higher inductive-inductive
type $\RR$ and the closeness relation. The code follows
\Cref{def:hott-reals} directly, except that positivity witnesses such as
$0 < \varepsilon$ and Cauchy approximation witnesses are carried explicitly
throughout, whereas the informal presentation of \Cref{chapter:reals-math}
routinely suppresses them. Cubical Agda's native support for higher inductive
types allows the mutually dependent definitions of $\RR$ and $\closesym$,
including the path constructor, to be expressed directly;
\Cref{fig:reals-definition} shows the definition as it appears in the code.

\begin{figure}[t!]
\begingroup\setstretch{1.0}\makeatletter\renewcommand{\verbatim@font}{\juliamono\footnotesize}\makeatother
\begin{verbatim}
data ℝ : Type

data Close : (ε : ℚ) → (0 < ε) → ℝ → ℝ → Type ℓ-zero

syntax Close ε p x y = x ∼[ ε , p ] y

CauchyApproximation : ((ε : ℚ) → 0 < ε → ℝ) → Type ℓ-zero
CauchyApproximation x =
  ((δ ε : ℚ) (p : 0 < δ) (q : 0 < ε) →
   x δ p ∼[ δ + ε , 0<+' {x = δ} {y = ε} p q ] x ε q)

data ℝ where
  rational : ℚ → ℝ
  limit : (x : (ε : ℚ) → 0 < ε → ℝ) →
          CauchyApproximation x →
          ℝ
  path : (x y : ℝ) →
         ((ε : ℚ) (p : 0 < ε) → x ∼[ ε , p ] y) →
         x ≡ y

data Close where
  rationalRational :
    (q r ε : ℚ) (φ : 0 < ε) →
    - ε < q - r → q - r < ε →
    rational q ∼[ ε , φ ] rational r
  rationalLimit :
    (q ε δ : ℚ) (φ : 0 < ε) (ψ : 0 < δ) (θ : 0 < ε - δ)
    (y : (ε : ℚ) → 0 < ε → ℝ) (ω : CauchyApproximation y) →
    rational q ∼[ ε - δ , θ ] (y δ ψ) →
    rational q ∼[ ε , φ ] (limit y ω)
  limitRational :
    (x : (ε : ℚ) → 0 < ε → ℝ) (φ : CauchyApproximation x)
    (r ε δ : ℚ) (ψ : 0 < ε) (θ : 0 < δ) (ω : 0 < ε - δ) →
    (x δ θ) ∼[ ε - δ , ω ] rational r →
    limit x φ ∼[ ε , ψ ] rational r
  limitLimit :
    (x y : (ε : ℚ) → 0 < ε → ℝ)
    (φ : CauchyApproximation x) (ψ : CauchyApproximation y)
    (ε δ η : ℚ) (θ : 0 < ε) (ω : 0 < δ) (π : 0 < η)
    (ρ : 0 < ε - (δ + η)) →
    (x δ ω) ∼[ ε - (δ + η) , ρ ] (y η π) →
    limit x φ ∼[ ε , θ ] limit y ψ
  squash :
    (ε : ℚ) (φ : 0 < ε) (u v : ℝ) →
    isProp $ u ∼[ ε , φ ] v
\end{verbatim}
\endgroup
\caption{The higher inductive-inductive definition of the HoTT book reals and
  the closeness relation in Cubical Agda, corresponding to
  \Cref{def:hott-reals}. The forward declarations of
  $\RR$ and \texttt{Close} establish the mutual dependence before the
  constructors are given.}
\label{fig:reals-definition}
\end{figure}

In addition to general $(\RR, \closesym)$-induction
and its specializations $\RR$-induction and $\closesym$-induction of
\Cref{def:general-induction,def:rr-induction,def:close-induction}, the
\texttt{Induction} module also formulates proposition-valued variants that
discharge the path hypotheses automatically, as well as a binary
proposition-valued variant for motives indexed by two reals. The enhanced
$(\RR, \closesym)$-recursion principle of \Cref{def:enhanced-recursion} is also
defined here. Because Cubical Agda has native support for higher inductive
types, the computation rules corresponding to all of these principles hold
definitionally. The \texttt{Definitions} module collects auxiliary definitions
used by the first three sections, including the dependent Cauchy approximation
conditions of \Cref{section:definition-induction-recursion}, the Lipschitz,
non-expanding, and continuity conditions of \Cref{section:lifting} and the
triangle inequality and roundedness conditions of \Cref{section:closeness}.

The code for results about closeness, corresponding to \Cref{section:closeness},
is split across three submodules. The first, \texttt{ReflexiveSymmetric},
establishes reflexivity and symmetry of the closeness relation
(\Cref{lem:close-reflexive,lem:close-symmetric}). The second,
\texttt{CloseAlternative}, contains the construction of the alternative
closeness relation $\approx$ of \Cref{theorem:alternative-closeness}. At nearly
1,850 lines, this is the single largest proof in the formalization, despite
receiving a relatively compact treatment in \Cref{section:closeness}. The three
submodules are separated primarily because iteratively type-checking the
alternative closeness proof became a bottleneck during development; splitting it
into its own module kept the feedback loop manageable. The third submodule,
\texttt{Other}, proves the equivalence of the alternative closeness relation
with the inductively defined closeness relation and records the resulting
roundedness and triangle inequality properties. The \texttt{Other} module also
includes helpers that package the interaction between closeness and limits. For
example, \texttt{closeLimit} derives $u \close{\varepsilon + \delta} \limit(y)$
from $u \close{\varepsilon} y_{\delta}$, supplying the necessary hypotheses for
the relevant closeness constructor in both cases of an induction argument. These
helpers are used throughout the algebra and order development whenever a proof
needs to argue that an arbitrary real is close to the limit of a Cauchy
approximation.

The code for \Cref{section:lifting} spans four modules. The
\texttt{Lipschitz.Base} module proves unary Lipschitz extension
(\Cref{lemma:lift-lipschitz}) using the enhanced $(\RR, \closesym)$-recursion
principle, while the \texttt{Nonexpanding} module proves binary non-expanding
extension (\Cref{lemma:lift-nonexpanding}). The \texttt{Lipschitz.Closed} module
formalizes Exercise~11.8 of the HoTT book
\cite{univalentfoundationsprogram2013}, which extends the Lipschitz extension
principle to maps defined on closed intervals of the rationals. We did not end
up using this result, since we decided to follow Gilbert's approach for
multiplication and reciprocal rather than continuing with the squaring-based
construction given in the HoTT book, which makes use of this exercise. However,
it may still be of independent interest for future work. The \texttt{Properties}
module contains the uniqueness of continuous extensions
(\Cref{lemma:continuous-extension-unique}) and its binary and ternary analogues
(\Cref{lemma:continuous-extension-unique2}), the continuous extension law
helpers at arities one through three
(\Cref{corollary:continuous-extension-law}), and the binary Lipschitz
composition lemma (\Cref{lemma:lipschitz-binary-composition}). These results are
the main tools used to transfer algebraic identities from $\QQ$ to $\RR$ in
\Cref{section:algebra-order}. The \texttt{Properties} module also contains the
proof that $\RR$ is a set (\Cref{corollary:reals-set}), the surjectivity of the
$\limit$ constructor, and the injectivity of the $\rational$ constructor, which
logically belong to earlier sections but are placed here due to module
dependency constraints.

The code for \Cref{section:algebra-order} is the largest portion of the
formalization, spanning approximately 6,200 lines across thirteen modules split
between algebra and order. On the algebra side, the \texttt{Negation} module
lifts rational negation via Lipschitz extension, and the \texttt{Addition}
module lifts rational addition via non-expanding extension. Their algebraic laws
are transferred from $\QQ$ using the continuous extension laws, and the
resulting structure is packaged as an Abelian group. The \texttt{Lattice} module
similarly lifts $\min$ and $\max$ via non-expanding extensions and establishes
their lattice structure. The \texttt{Multiplication} module follows the iterated
Lipschitz extension strategy described in \Cref{section:algebra-order}, and the
needed algebraic identities are lifted to yield a commutative ring instance. The
\texttt{Reciprocal} module follows the same local-to-global pattern to obtain
the reciprocal on positive reals, and then extends to all reals apart from zero
by case splitting on the apartness witness. On the order side, the \texttt{Base}
module defines non-strict and strict inequality as in
\Cref{section:algebra-order}. The \texttt{Magnitude} and \texttt{Distance}
modules define absolute value and the distance metric on $\RR$ and and develop
their basic theory. The remaining order modules establish interactions between
order and algebraic operations: monotonicity and order-reflection of addition
with respect to both $\le$ and $<$, weak linearity, Gilbert's alternative
characterization of strict inequality
(\Cref{theorem:less-than-iff-perturb-less-equal}), and positivity preservation
and strict monotonicity for multiplication by positive reals.

\section{On the Use of Claude Code}
\label{section:claude-code}

We experimented with Claude Code, Anthropic's coding agent~\cite{claudecode}. We
used Claude Code to aid navigation of the Cubical Agda standard library, and as
a tool for converting detailed proof sketches into formal Agda proofs.

The Cubical Agda standard library is large, and locating the right lemma or
definition by name alone can be difficult. Because Claude Code has tool access
to the file system, it could search the library by description rather than by
name, explain unfamiliar definitions, and help navigate module structure. This
made it an effective form of semantic search during development.

\begin{figure}[t!]
\begingroup\setstretch{1.0}\makeatletter\renewcommand{\verbatim@font}{\juliamono\small}\makeatother
\begin{verbatim}
Can you fill in the type hole for the χ subterm in
`maxMultiplyBoundedReciprocalPositiveContinuous` in Reciprocal.agda based on the
following proof sketch?

- By `∣∣≤rational` (Properties2.agda), there exists an L : Q
  with |max(x, δ)| ≤ L
- By the continuity of `boundedReciprocalPositive`, there is
  a θ so that x ~[θ] y ⟹
  brp(δ, _, x) ~[(L+1)⁻¹ ε/2] brp(δ, _, y)
- Choose η₂ := min(θ, 1)
- Fix y : R and assume x ~[η₂] y
- Then x ~[1] y. By `maxNonexpandingℝ₂`, we have
  max(x, δ) ~[1] max(y, δ). Applying `close→≤+ε` to
  |x| ≤ L we obtain |max(y, δ)| ≤ L + 1
- By `·distanceₗ`,
  |max(y, δ) · brp(δ, x) - max(y, δ) · brp(δ, y)| = |max(y, δ)| · |brp(δ, x) - brp(δ, y)|
- Then
  |max(y, δ)| · |brp(δ, x) - brp(δ, y)| ≤ (L + 1) · (L + 1)⁻¹ · ε/2 = ε/2
- Then using `close→distance<`, we obtain the required
  closeness result

Hint: For the very last step, you'll probably need to use
the pattern for substituting the propositional positivity
argument to `Close` used in the π term below.
\end{verbatim}
\endgroup
\caption{A prompt given to Claude Code during the development of the reciprocal
  construction. The sketch specifies the proof strategy step by step, naming
  intermediate results and the library lemmas to apply. We abbreviate
  \texttt{boundedReciprocalPositive} as \texttt{brp} to fit the page.}
  \label{fig:claude-code-prompt}
\end{figure}

More surprisingly, Claude Code could take informal proof sketches and produce
formal Agda proof terms that type-checked, requiring only minor cleanup
afterwards. \Cref{fig:claude-code-prompt} shows a representative prompt from the
development. The proof sketch specifies the proof strategy step by step, naming
the intermediate results and the specific library lemmas to invoke at each
stage. The sketch also includes a hint about a substitution pattern for
rewriting a dependent positivity witness in the final step. Given this level of
detail, Claude Code produced a proof term that, after small adjustments for code
style, is essentially the one that remains in the codebase as
\texttt{maxMultiplyBoundedReciprocalPositiveContinuous} in
\texttt{Reciprocal.agda}. The proof strategy had to be specified in the prompt,
but the agent handled the translation to Agda syntax, including the selection of
correct library functions, the bookkeeping of positivity witnesses, and the
assembly of the proof term.

\section{Lessons Learned}
\label{section:lessons-learned}

The effort of formalization surfaced three insights about the interaction
between the informal presentation of the HoTT book reals and their realization
in Cubical Agda. In each case, the precision required by formalization surfaced
something that informal reading alone did not. We describe them in turn and then
identify the patterns they share.

The first concerns the alternative closeness relation of
\Cref{theorem:alternative-closeness}. The constructors of $\closesym$ provide
sufficient conditions for each combination of a rational and a limit, but not
necessary conditions. Given $\limit(x) \close{\varepsilon} \rational(r)$ as a
hypothesis, for instance, nothing about the inductive-inductive definition lets
us immediately extract a $\delta : \Qp$ satisfying the condition of the
limit-rational constructor. Indeed, as the HoTT book notes, inductive type
families ``have a tendency to contain `more than was put into them' ''
\cite[\S~11.3.2]{univalentfoundationsprogram2013}. The book flags this
explicitly, resolving it by defining a relation $\approx$ by recursion which
computes on construtors and proving that $\approx$ and $\closesym$
coincide. This resolution was already in place before the formalization began.
What the effort of formalization added was a much slower understanding of
\emph{why} this construction is the right response. The
\texttt{CloseAlternative} module contains the longest single proof in the
codebase, and only after spending time on it and then applying the alternative
closeness relation in later proofs did its role become clear. In principle, each
extraction of information from a closeness hypothesis could be carried out
locally by an induction argument; indeed, that is exactly how the forward
implication
\[
  u \close{\varepsilon} v \implies u \approx_{\varepsilon} v
\]
of \Cref{theorem:close-eq-close'} is proved. The alternative closeness relation
packages the necessary conditions once, as a globally applicable statement, so
that downstream proofs can invoke a single reusable statement rather than
reinvoke induction locally each time. This role was not immediately evident from
reading the book alone; it became evident only after formalizing the equivalence
and applying the relation in subsequent proofs.

The second insight concerns the enhanced $(\RR, \closesym)$-recursion principle
of \Cref{def:enhanced-recursion}. The HoTT book states the limit case hypothesis
in a notation that writes $f(x_{\varepsilon})$ in the inductive hypothesis,
where $f$ is the function being defined by recursion
\cite[\S~11.3.2]{univalentfoundationsprogram2013}. Both the approximation $x$
itself and the inductively defined values $f(x_{\varepsilon})$ are therefore
tacitly available. Our initial, naive transcription into Agda, however, dropped
$x$ from the inductive hypothesis entirely. In the induction principle, the
motive $A : \RR \to \mathcal{U}$ is indexed by $\RR$, so $x_{\varepsilon}$
appears naturally in the type $A(x_{\varepsilon})$ of the inductively assigned
values. In the non-dependent recursion principle, the motive is a plain type
$A : \mathcal{U}$, and no mention of $x_{\varepsilon}$ remains. It is natural to
drop $x$ from the inductive hypothesis altogether, leaving only $f : \Qp \to
A$. This suffices for most uses of the principle, because a typical recursion is
concerned with producing output in the codomain $A$ and has no occasion to refer
back to the domain approximation. The weakness only surfaced in the inner
recursion used to define alternative closeness, where the codomain is itself a
family of relations on $\RR$. In the limit-limit case there, the body of the
required existential must name a real drawn from the domain
approximation. Without access to $x$ in the recursion hypothesis, the case
cannot even be formulated. We wrote more than two thousand lines of code using
the weaker recursion principle before reaching this point. Repairing it required
reformulating the recursion principle to carry both the domain approximation and
its codomain assignment (see the comment at lines 310--325 of
\texttt{Induction.agda}) and refactoring each existing call site, including the
Lipschitz extension in \texttt{Lipschitz.Base}. Here the formalization did not
merely deepen understanding of an existing construction; it required us to
notice that the naive transcription had silently dropped information that the
informal presentation carried for free.

The third insight concerns the extension of continuous functions of several
variables. The HoTT book states the uniqueness of continuous extensions for maps
of one variable (\Cref{lemma:continuous-extension-unique}) and remarks in
passing that identities in several variables can be extended in the same
manner. The precise hypotheses are left unstated. In particular, it is not
obvious upon first reading whether joint continuity is required or whether
continuity in each variable separately suffices. To transfer algebraic
identities from $\QQ$ to $\RR$ throughout \Cref{section:algebra-order}, we
needed to commit to one or the other. We showed in
\Cref{lemma:continuous-extension-unique2} that for binary functions, separate
continuity in each variable is sufficient, with the proof proceeding by applying
the one-variable lemma one coordinate at a time. We formulate a ternary analogue
similarly in the code. Here the formalization genuinely supplied content that
the informal presentation only gestures at. The informal text does not say the
wrong thing, but it does not say enough to be applied directly, and making the
result applicable in the formalization required articulating hypotheses and
proving a new lemma.

These three instances illustrate two recurring patterns in the relationship
between informal mathematics and its formalization. In the first case, the
informal presentation is complete and rigorous, but formalizing it made its
structure evident in a way reading alone did not. In the other two, the formal
system required explicit articulation of information that the informal
presentation left implicit, and supplying that information became part of the
mathematical work of the formalization.

\chapter{Related Work}
\label{chapter:related-work}

The present chapter discusses the prior and concurrent work most directly
related to this thesis, including other formalizations of the real numbers and
broader programs in constructive analysis within univalent foundations.

The closest prior work is Gilbert's 2017 formalization of the HoTT book reals in
the Rocq proof assistant \cite{gilbert2017}. Gilbert, like the HoTT book,
follows \textcite{oconnor2007} in defining closeness as a family of binary
relations indexed by positive rationals, rather than as a real-valued distance
function. Beyond the formalization itself, Gilbert generalizes the HoTT book
construction to arbitrary premetric spaces, showing that Cauchy completion is a
monad on the category of premetric spaces with Lipschitz functions, and uses
Altenkirch, Danielsson, and Kraus's partiality
monad~\cite{altenkirchdanielssonkraus2017} to give a semi-decision procedure
comparing a real to a rational. Rocq lacks native support for both
inductive-inductive types and higher inductive types; Sozeau's experimental
branch added support for inductive-inductive types, but higher inductive types
remain unsupported, requiring Gilbert to axiomatize the path constructors of the
HoTT book reals. This motivates our choice of Cubical Agda, which natively
supports both. We adopt Gilbert's local-to-global strategy for multiplication
and reciprocal and his alternative characterization of strict inequality, and
depart in two respects. Cubical Agda supplies the path constructors and their
computation rules definitionally, and we use Kraus's \cite{kraus2015} theorem on
maps out of propositional truncations into sets in place of Gilbert's definition
by surjection to package the local-to-global step.

Concurrently, \textcite{molenaturekgrzybowskiborsetto2026} have independently
developed a Cubical Agda formalization of the HoTT book reals, pursued as an
in-progress pull request to the Cubical Agda standard library. They
additionally contribute algebraic and premetric infrastructure,
including pseudolattices, ordered commutative rings, and Archimedean rings
adapted from the HoTT book's ordered Heyting fields. In ongoing work, they
generalize Gilbert's premetric spaces from positive rationals to the positive
cone of an arbitrary ordered commutative ring. Their goals extend
into analysis beyond the scope of this thesis, including the Riemann integral,
the fundamental theorem of calculus, the mean value theorem, and trigonometric
functions.

The HoTT book construction can be read as an adaptation of earlier work by
\textcite{richman2008}, who developed a notion of completion for premetric
spaces that does not require countable choice. The HoTT book's closeness
relation and limit constructor apply Richman's completion strategy to the
rationals within univalent foundations.

\textcite{booij2020} pursues constructive analysis in univalent type theory.
Booij surveys several constructions of the real numbers, including an
alternative characterization of the HoTT book reals as a homotopy-initial Cauchy
structure, and proves equivalences among them. Booij's central contribution is
the notion of a \emph{locator}, extra structure attached to reals that enables
discrete observations such as signed-digit expansions while remaining compatible
with the extensionality principles of univalent type theory. Using locators,
Booij develops intermediate value theorems, Riemann integration, and the
Picard-Lindelöf theorem. The work is pencil-and-paper mathematics accompanied
by a Haskell prototype of locators. A chapter addresses strategies for
proof-assistant implementation but stops short of formalization. We return to
locators in \Cref{chapter:future-work}, where we discuss generalizing them
beyond the real numbers to a broader class of structures.

\textcite{murray2022} formalized Bishop-style constructive real numbers in
non-cubical Agda, including their arithmetic, ordering, Cauchy completeness, and
uncountability. In his future work, Murray suggested a Cubical Agda port as a
possible next step, while noting that HoTT-based definitions of the real numbers
might be preferable to Bishop-style ones, rendering such a port redundant. Our
work here takes up exactly that suggestion, replacing the Bishop-style setoid
construction with the HoTT book reals as a higher inductive-inductive type in
Cubical Agda.

\chapter{Future Work}
\label{chapter:future-work}

We identify three directions for future work arising from this formalization,
concerning the algebraic packaging of the HoTT book reals, the extension
principles used to lift operations from the rationals, and the generalization of
locators to a broader class of structures.

Our formalization stops just short of characterizing the HoTT book reals under
their universal property, namely that they are the initial Cauchy complete
Archimedean ordered field. We cover up to the end of Section 11.3.3 of the HoTT
book, concluding with Theorem 11.3.48 which states that the HoTT book reals form
an Archimedean ordered field~\cite{univalentfoundationsprogram2013}. Section
11.3.4 comprises only two further theorems: Theorem 11.3.49 shows Cauchy
completeness and Theorem 11.3.50 shows initiality. We stopped short not because
these theorems are difficult to prove, but because the Cubical Agda standard
library lacked support for a constructive formulation of fields at the time of
writing. As a consequence, our formalization does not contain a unified
statement corresponding to
\Cref{theorem:reals-archimedean-ordered-field}. Instead, our results are
scattered across each of the properties which make up an Archimedean ordered
field definition, although we have verified each constituent property
individually, so that the proof will follow immediately once a suitable
definition is in place. A natural next step would be to add these definitions to
the Cubical Agda standard library, allowing the universal property to be stated
formally.

We mentioned in \Cref{section:lifting} that the Lipschitz and non-expanding
extension principles of \Cref{lemma:lift-lipschitz,lemma:lift-nonexpanding} were
not to be read as the final word on the conditions under which maps on the
rationals extend to the reals. The constructions of multiplication and
reciprocal in \Cref{section:algebra-order} required iterated Lipschitz
extensions followed by elimination of the chosen bound via Kraus's theorem,
adding considerable complexity compared to the constructions of negation,
addition, and the lattice operations. A broader extension principle for
uniformly continuous maps would simplify at least the local extension step,
since bounded multiplication is uniformly continuous as a map of two rational
variables. Booij suggests that such an extension ought to be possible
\cite[Chapter~10]{booij2020}.

The path constructor for the HoTT book reals ensures that two reals which are
arbitrarily close are identified, so that the natural notion of equality among
reals coincides with the identity type. This avoids the drawbacks of a
Bishop-style setoid construction, where the equivalence relation and the
identity type diverge. However, this identification is also exactly what
prevents us from extracting rational approximations. Booij addresses this by
introducing the concept of a locator, extra structure on a real number that
restores the ability to compute with approximations~\cite{booij2020}. A locator
is obtained by strengthening the locatedness property of Dedekind cuts from a
property to structure. At the end of his chapter on metric spaces, Booij writes:
``We have not found a definition of locators for general metric spaces, of which
the locators introduced in Chapter 6 would be a special case. This would be an
important future direction of research'' \cite[p.~183]{booij2020}. We have
developed a construction of the HoTT book reals in Cubical Agda, where our
proofs are type-checked for correctness. Generalizing locators to a larger class
of structures would allow the same framework that verifies the correctness of
the mathematical theory to also execute computations on analytic objects via
term normalization.

\chapter{Conclusion}
\label{chapter:conclusion}

This thesis presented a formalization of the HoTT book reals in Cubical Agda. We
covered the higher inductive-inductive definition, the closeness relation and
its alternative characterization, extension principles for Lipschitz and
non-expanding maps, and the algebraic and order-theoretic structure culminating
in the proof that the HoTT book reals form an Archimedean ordered field. The
development is open source at \url{https://github.com/utahplt/hott-reals} and
typechecks without postulates or holes.

The process of formalization required us to make explicit mathematical structure
that the informal presentation of the HoTT book left implicit, including the
precise form of the enhanced recursion principle and the hypotheses needed for
multi-variable continuous extensions. In several cases, the precision demanded
by the type checker surfaced insights that were not apparent from reading the
informal account alone.

This development clarifies how the HoTT book reals behave in a proof assistant
with native support for higher inductive types. It provides a foundation for
further machine-assisted work in constructive analysis.

\printbibliography[heading=bibintoc]

\newpage
\thispagestyle{empty}

\vspace*{90pt}
\begin{center}

    Name of Candidate: Jackson Brough\\
    Date of Submission: April 17, 2026
    
\end{center}

\end{document}